\documentclass[11pt]{article}
\usepackage{graphicx}
\usepackage{citesort}
\usepackage{amssymb}
\usepackage{amsfonts}
\usepackage{amsmath}
\usepackage{epsfig}
\usepackage{color}
\usepackage{fancybox}
\usepackage{textcomp}
\usepackage{multirow}
\usepackage{appendix}
\usepackage{setspace}

\renewcommand{\baselinestretch}{1.65}
\newtheorem{Theorem}{Theorem}

\newtheorem{Lemma}{Lemma}
\newtheorem{Assumption}{Assumption}

\newtheorem{Remark}{Remark}

    \newcommand{\qa}{{\bf a}}
    \newcommand{\qb}{{\bf b}}
    \newcommand{\qc}{{\bf c}}
    
    \newcommand{\qe}{{\bf e}}

    \newcommand{\qv}{{\bf v}}
    
    \newcommand{\qx}{{\bf x}}
    \newcommand{\qy}{{\bf y}}

    \newcommand{\qA}{{\bf A}}
    \newcommand{\qB}{{\bf B}}
    \newcommand{\qC}{{\bf C}}
    \newcommand{\qD}{{\bf D}}
    \newcommand{\qE}{{\bf E}}
    
    \newcommand{\qG}{{\bf G}}
    \newcommand{\qH}{{\bf H}}
    \newcommand{\qI}{{\bf I}}

    \newcommand{\qR}{{\bf R}}
    \newcommand{\qS}{{\bf S}}
    \newcommand{\qT}{{\bf T}}
    \newcommand{\qU}{{\bf U}}
    \newcommand{\qV}{{\bf V}}
    
    \newcommand{\qX}{{\bf X}}

    \newcommand{\qzero}{{\bf 0}}

    \def\ang#1{\mbox{$\langle #1 \rangle$}}

    \def\Complex{{\rm\rule[.23ex]{.03em}{1.1ex}\kern-.3em{C}}}
    
    \def\eqover#1{\mbox{$\buildrel #1 \over =$}}
    \def\leqover#1{\mbox{$\buildrel #1 \over \leq$}}
    \def\Imag#1{\mbox{${\Im}\left\{ #1 \right\}$}}
    \def\dist{\mbox{${\sf dist}(z, {\mathbb R}^+)$}}
    \def\largeN{\mbox{$\boldsymbol{\cal N}$}}

    \newcommand{\be}{\begin{equation}} \newcommand{\ee}{\end{equation}}
    \newcommand{\bea}{\begin{eqnarray}} \newcommand{\eea}{\end{eqnarray}}
    \newcommand{\benum}{\begin{enumerate}} \newcommand{\eenum}{\end{enumerate}}

    \newcommand{\bbR}{{\mathbb R}}
    \newcommand{\bbC}{{\mathbb C}}
    \newcommand{\bbS}{{\mathbb S}}

    \newcommand{\calV}{{\mathcal V}}
    \newcommand{\calX}{{\mathcal X}}

    \newcommand{\calqB}{\boldsymbol{\cal B}}
    \newcommand{\calqX}{\boldsymbol{\cal X}}

    \newcommand{\diag}{{\sf diag}}
    
    \newcommand{\tr}{{\sf tr}}
    \newcommand{\Ex}{{\sf E}}
    \newcommand{\rank}{{\sf rank}}

    \newcommand{\sfj}{{\sf j}}
    \newcommand{\sfF}{{\sf F}}

    \newcommand{\qDelta}{{\boldsymbol \Delta}}
    \newcommand{\qPsi}{{\boldsymbol \Psi}}
    \newcommand{\qPhi}{{\boldsymbol \Phi}}
    \newcommand{\qXi}{{\boldsymbol \Xi}}
    \newcommand{\qTheta}{{\boldsymbol \Theta}}
    
    \newcommand{\qSigma}{{\boldsymbol \Sigma}}
    \newcommand{\qGamma}{{\boldsymbol \Gamma}}
    \newcommand{\qUpsilon}{{\boldsymbol \Upsilon}}

    \newcommand{\qeta}{{\boldsymbol \eta}}
    
    \newcommand{\qkappa}{{\boldsymbol \kappa}}
    \newcommand{\qtau}{{\boldsymbol \tau}}

\begin{document}

\setcounter{page}{0}

\title{\LARGE\bf A Deterministic Equivalent for the Analysis of\\
Non-Gaussian Correlated MIMO Multiple Access Channels\thanks{This work was supported in part by National Science Council, Taiwan, under grant
NSC100-2221-E-110-052-MY3.}}

\author{Chao-Kai Wen\footnote{Institute of Communications Engineering, National Sun Yat-sen University, Kaohsiung 804, Taiwan.},
~Guangming Pan\footnote{School of Physical and Mathematical Sciences, Nanyang Technological University, Singapore.}, ~Kai-Kit Wong\footnote{Department of
Electronic and Electrical Engineering, University College London, London, WC1E 7JE, United Kingdom.}, ~Mei-Hui Guo\footnote{Department of Applied Mathematics,
National Sun Yat-sen University, Kaohsiung 804, Taiwan.},~and~Jung-Chieh Chen\footnote{Corresponding author. E-mail: {\sf jc\underline{ }chen@ieee.org}. Return
Address: Department of Optoelectronics and Communication Engineering, National Kaohsiung Normal University, Kaohsiung 802, Taiwan. Tel: +886 7 717-2930 Ext 7718.}}

\date{}
\maketitle \thispagestyle{empty}

\begin{abstract}
Large dimensional random matrix theory (RMT) has provided an efficient analytical tool to understand multiple-input multiple-output (MIMO) channels and to aid the
design of MIMO wireless communication systems. However, previous studies based on large dimensional RMT rely on the assumption that the transmit correlation matrix
is diagonal or the propagation channel matrix is Gaussian. There is an increasing interest in the channels where the transmit correlation matrices are generally
nonnegative definite and the channel entries are non-Gaussian. This class of channel models appears in several applications in MIMO multiple access systems, such
as small cell networks (SCNs). To address these problems, we use the generalized Lindeberg principle to show that the Stieltjes transforms of this class of random
matrices with Gaussian or non-Gaussian independent entries coincide in the large dimensional regime. This result permits to derive the deterministic equivalents
(e.g., the Stieltjes transform and the ergodic mutual information) for non-Gaussian MIMO channels from the known results developed for Gaussian MIMO channels, and
is of great importance in characterizing the spectral efficiency of SCNs.

\begin{center}
{\bf Index Terms}
\end{center}
\vspace{-.05in} Generalized Lindeberg principle, Interpolation trick, Large dimensional RMT, MIMO, Small cell networks, Shannon transform, Stieltjes transform.
\end{abstract}

\thanks{\singlespacing}

\newpage
\section*{\sc I. Introduction}
The seminal works by Foschini {\em et al.}~\cite{Fos-98} and Telatar \cite{Tel-99} have inspired the world to realize the huge capacity of multiple-input
multiple-output (MIMO) antenna systems and shed light on the capacity-achieving strategies of such systems. However, exact analysis for the achievable rates of
MIMO channels could be difficult and for some channel models unsolvable. In the last few years, large-system approaches have emerged as a means to circumvent the
mathematical difficulties, greatly motivated by the landmark contributions of Verd\'u-Shamai \cite{Ver-99} and Tse-Hanly \cite{Tse-99} using large dimensional
random matrix theory (RMT) to various problems in information theory. Since then, a large body of performance analyses of various MIMO channels were obtained by
large dimensional random matrix tools such as the Stieltjes transform method (or the Silverstein-Bai method) \cite{Sil-95},\footnote{In recent years, this method
due to Silverstein and Bai has been developed into a much useful tool, widely known as the Stieltjes transform method in the spectral analysis of large dimensional
random matrices.} the Gaussian tools (integration by part and the Poincar\'e-Nash inequality) \cite{Pas-05}, the free probability \cite{Voi-97}, and the replica
method \cite{Edw-75}. See \cite{Tul-04,Cou-11,Bai-09} for more details.

For channel matrices with Gaussian entries, the replica method, an approach originally developed in statistical physics, serves as a powerful tool to derive the
relevant results. For example, it has been used to obtain asymptotic mutual information results for Rayleigh \cite{Mou-03} and Rician fading \cite{Tar-08} channels
with separately correlated antennas. Nevertheless, this method is mathematically incomplete, to say the least. To acquire a more sound mathematical procedure,
advanced tools such as the Gaussian tools and the Stieltjes transform method are required. Using the Gaussian tools, the asymptotic mutual information expressions
for Rayleigh and Rician fading channels have been confirmed rigorously by Hachem {\em et al.}~\cite{Hac-08} and Dumont {\em et al.}~\cite{Dum-10}, respectively.
Based on the Stieltjes transform method, Couillet {\em et al.}~recently studied a MIMO multiple access channel (MAC) with separately correlated user channels
\cite{Cou-09}. In this case, each user's channel matrix, $\qH_k$, can be written in the form $\qR_k^\frac{1}{2}\qX_k \qT_k^\frac{1}{2}$, where $\qX_k$ has
independent and identically distributed (i.i.d.) zero-mean Gaussian entries, and $\qR_k$ and $\qT_k$ are both deterministic nonnegative definite matrices which,
respectively, characterize the spatial correlation structure at the receiver and transmitter sides separately.

Though strictly speaking, the large-system results are only asymptotically tight, they provide reliable performance predictions even for small system dimensions
and at a much lower computational cost than Monte-Carlo simulations, as well as offer insightful understanding on communications channels. Moreover, large-system
results are also important for designing many practical wireless systems such as precoder design \cite{Art-09,Dup-10,Cou-09}, optimal training length design
\cite{Hoy-11c,Wag-09}, scheduling \cite{Huh-10}, and others \cite{Hoy-11b,Lak-10}. For most contributions, the elements of the MIMO channel matrix are assumed to
be multivariate Gaussian distributions; that is, the amplitudes of the channel fading coefficients are either Rayleigh or Rician distributed. Despite being the
most popular models for small-scale amplitude fading, there are more and more results to suggest different models \cite{Mol-05a,Foe-03,Mol-05}. For example,
\cite{Mol-05} proposed that Nakagami-$m$ distribution is best suited for modeling the small-scale amplitude fading in such as indoor residential/office, industrial
environments, and suburban-like microcell environments. In addition, the log-normal distribution has recently been used to describe the small-scale amplitude
fading in the IEEE 802.15.3a \cite{Foe-03}. There is clearly an increasing demand to investigate channels with non-Gaussian fading and their performance. Whether
the systems specifically designed for Gaussian scenarios can still work well in non-Gaussian environments is unknown, and the results available in the literature
so far are too limited to answer this question \cite{Cou-09,Hac-07,Kam-10}.

To appreciate the objective of this paper, it is important to understand the limitations of the existing results for non-Gaussian channels. In \cite{Cou-09}, the
results were only derived under the assumption that each transmitter-side correlation matrix, $\qT_k$, is diagonal, although it was conjectured that the results
might be valid even when $\qT_k$ is nonnegative definite. A channel model composed of a general variance profile and a deterministic line-of-sight (LOS) component
was studied in \cite{Hac-07} which partially generalized the results in \cite{Cou-09}. However, as compared to \cite{Cou-09}, the matrices $\qR_k$'s in
\cite{Hac-07} cannot be nonnegative definite.

\begin{figure}
\begin{center}
\resizebox{4.5in}{!}{%
\includegraphics{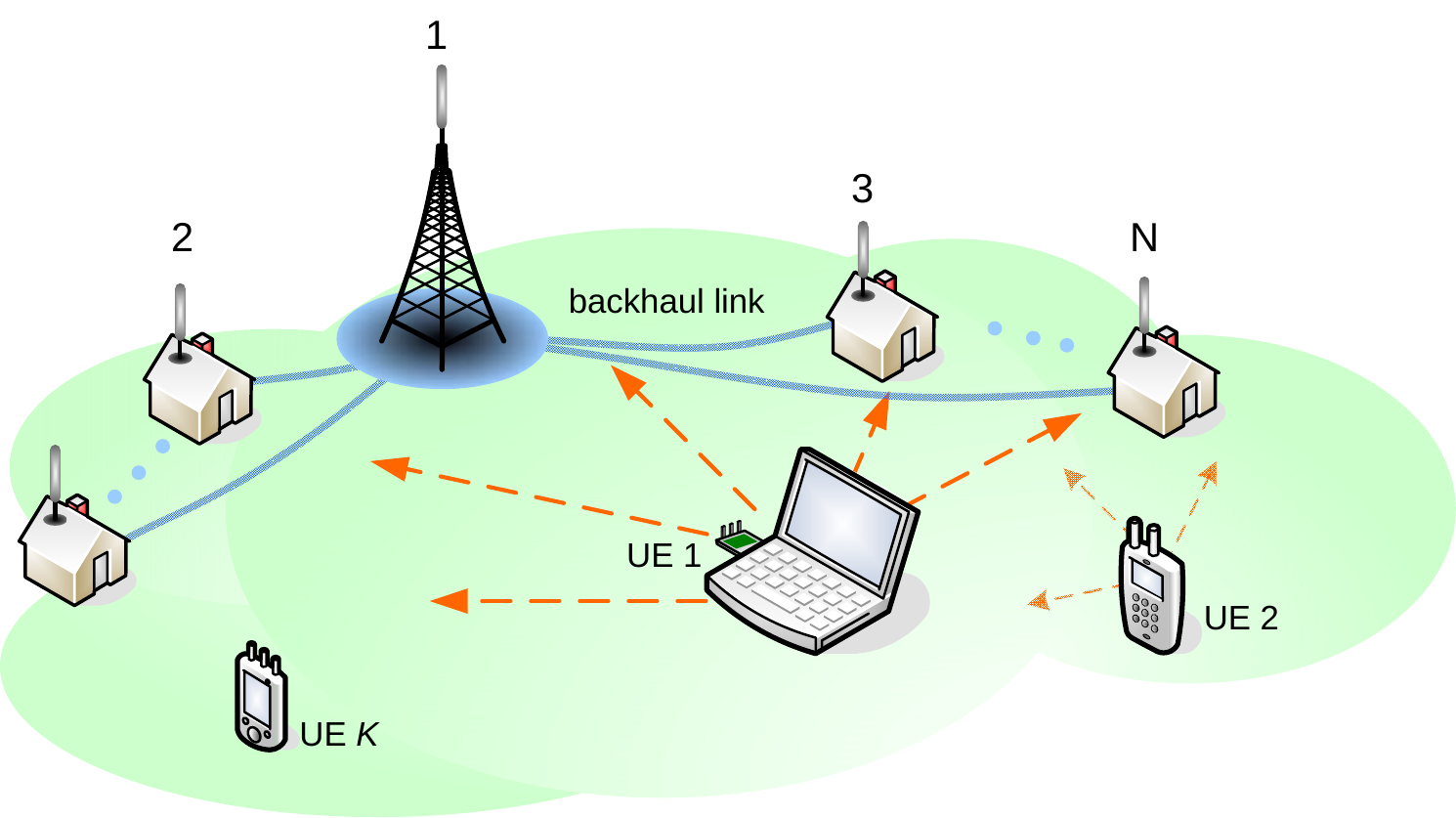} }%
\caption{A small-cell network.}\label{fig:SCN}
\end{center}
\end{figure}

This paper aims to extend previous large-system results to a more general class of random matrices with non-Gaussian entries. As in \cite{Cou-09}, we consider a
$K$-user MIMO MAC, in which each $\qH_k$ is spatially correlated separately at both sides. In our model, a deterministic LOS component $\bar\qH_k$ is also
considered. More specifically, the concerned Kronecker channel $\qR_k^\frac{1}{2} \qX_k \qT_k^\frac{1}{2}$ can be described as follows. The entries of $\qX_k$ are
i.i.d. complex centered random variables ({\em not} necessarily Gaussian), $\qT_k$'s are deterministic nonnegative definite matrices, and $\qR_k$'s are diagonal
nonnegative matrices. This model arises in small-cell networks (SCNs) as shown in Figure \ref{fig:SCN}. The SCNs, which are typically composed of densely deployed
low-cost low-power base stations (BSs), have attracted considerable attention for their potential to increase the capacity of cellular networks
\cite{Hoy-11,Cou-09,Wag-09}. In these networks, the channel fading would tend to be non-Gaussian. In contrast to \cite{Hac-07}, our consideration allows user
equipments (UEs) to be equipped with multiple spatially correlated antennas, which is a typical phenomenon due to space limitation of UEs.

There are several obstacles when one intends to apply the Stieltjes transform method originally developed for the case with diagonal $\qT_k$ (e.g.,
\cite{Sil-95,Cou-09}) to that with general nonnegative definite $\qT_k$ \cite{Zha-06}. To overcome the difficulties, using the generalized Lindeberg principle
\cite{Cha-06,Kor-11}, we show that under very mild conditions, the Stieltjes transforms of the considered random matrices with Gaussian entries and that with
non-Gaussian entries coincide in the large dimensional regime. This result enables us to derive the deterministic equivalents (e.g., the Stieltjes transform and
the ergodic mutual information) for non-Gaussian MIMO channels from the known results for Gaussian MIMO channels. We therefore generalize the deterministic
equivalents of previous results to the SCNs. For uncorrelated channel matrices with i.i.d.~entries, such property is implicit in \cite[Figure 4]{Fra-08} from
computer simulations and has recently been proved in \cite[Corollary 2]{Kor-11}. However, in our derivation, we prove that the deterministic equivalents of the
MIMO MAC channel in \cite{Cou-09} are true even if the entries of $\qX_k$ are non-Gaussian, and those $\qR_k$ and $\qT_k$ are deterministic nonnegative definite
matrices.\footnote{Note that if the LOS is absent, we allow $\qR_k$'s to be nonnegative definite. See Section III for detail.} Therefore, we prove the conjecture
made in \cite{Cou-09} entirely.

The remainder of this paper is structured as follows. In Section II, we introduce the channel model of the SCNs. Section III then presents our main results and
outline their proofs whose details are given in the appendices. Some mathematical tools needed in proving the results are reviewed in Appendix D. Simulation
results are provided in Section IV and finally we conclude the paper in Section V.

Notations---Throughout this paper, the complex number field is denoted by $\mathbb{C}$. For any matrix $\qA\in \bbC^{N \times n}$, $A_{ij}$ denotes the $(i,j)$th
entry, while $\qA^T$, and $\qA^H$ return the transpose and the conjugate transpose of $\qA$, respectively. For a square matrix $\qB$, $\qB^\frac{1}{2}$,
$\qB^{-1}$, ${\sf tr}(\qB)$, and $\det(\qB)$ denote the principal square root, inverse, trace, determinant of $\qB$, respectively. Also, $\qI_N$ is the $N \times
N$ identity matrix, $\qzero_N$ denotes either the $N \times N$ zero matrix or a zero vector depending on the context, $\|\cdot\|$ represents the Euclidean norm of
an input vector or the spectral norm of an input matrix, $\|\cdot\|_{\sfF}$ denotes the Frobenius norm of a matrix, $\rho(\cdot)$ represents the spectral radius
(i.e., the largest absolute value of the eigenvalues) of a matrix, $\Ex \{\cdot\}$ returns the expectation of an input random entity, $\log(\cdot)$ is the natural
logarithm, $\Re\{\cdot\}$ and $\Im\{\cdot\}$ return the real part and the imaginary part of an input entity respectively, $1_{\cal A}$ denotes the indicator
function of the set ${\cal A}$, and $\otimes$ is the Kronecker product \cite{Gra-81}. We use $C$ (or $C_p, C', C'',\dots$) to denote a universal constant whose
value does not depend on matrix sizes but may vary from one appearance to another. Almost sure (a.s.) convergence is denoted by $\xrightarrow{a.s.}$. If $\{ a_i
\}_i$ is a sequence of real numbers, then $b_i = O(a_i)$ and $b_i = o(a_i)$ stands for $|b_i| \leq C |a_i|$ and $\lim_i\frac{b_i}{a_i} \rightarrow 0$ respectively.
As usual, ${\sf j} \equiv \sqrt{-1}$, $\bbR^+ \equiv \{x \in \bbR : x \geq 0 \}$, and $\bbR^- \equiv \{x \in \bbR : x\leq 0 \}$. Also, $\bbC^+ \equiv \{ z = z_1 +
{\sf j} z_2 \in \bbC : z_2 > 0 \}$ and $\bbC^-\equiv \{ z = z_1 + {\sf j} z_2 \in \bbC : z_2 <0 \}$.

\section*{\sc II. Channel Model and Problem Formulation}
\subsection*{A. Multi-cell MIMO-MAC with LOS and Spatial Correlation}
As shown in Figure \ref{fig:SCN}, we consider a MIMO-MAC system with $K$ UEs, labeled as ${\sf UE}_1,\dots,{\sf UE}_K$, which are equipped with $n_1,\dots,n_K$
antennas, respectively. The $K$ UEs transmit to $N$ interconnected small-cell single-antenna BSs simultaneously. In this paper, we use the Kronecker model to
characterize the spatial correlation of the MIMO channel for each link so that the correlation properties at the BS and any UE are modeled separately, e.g.,
\cite{Shi-00}. Specifically, user $k$'s channel, $\qH_k \in{\mathbb C}^{N\times n_k}$, can be written as
\begin{equation}\label{eq:Spatial_Cov}
\qH_k = \qR_k^\frac{1}{2} \qX_k \qT_k^\frac{1}{2} + \bar{\qH}_k,
\end{equation}
where $\qR_k = \diag(r_{k,1},\dots,r_{k,N})$ is a deterministic diagonal matrix with the $i$th diagonal entry, $r_{k,i}$, being the channel gain from ${\sf UE}_k$
to the $i$th receiving antenna (or the $i$th BS), $\qT_k \in\mathbb{C}^{n_k \times n_k}$ is a deterministic nonnegative definite matrix, which expresses the
correlation of the transmit signals across the antenna elements of ${\sf UE}_k$, $\qX_k \equiv [\frac{1}{\sqrt{n_k}} X_{ij}^{(k)}] \in\mathbb{C}^{N\times n_k}$
consists of the random components of the channel in which the elements $\{X_{ij}^{(k)}\}_{1\leq i \leq N; 1 \leq j \leq n_k}$ are i.i.d. complex random variables
with zero mean and variance of $P_k$, and $\bar\qH_k \in\mathbb{C}^{N\times n_k}$ is a deterministic matrix corresponding to the LOS of the channel.

To get a proper definition on the signal-to-noise ratio (SNR), consider the power of the channel:
\begin{equation}
\Ex \left[\tr\left(\qH_k\qH_k^H\right)\right] = \frac{P_k}{n_k}\tr\left(\qR_k\right)\tr\left(\qT_k\right)+\tr\left(\bar{\qH}_k\bar{\qH}_k^H\right).
\end{equation}
It is customarily assumed that $\qR_k$, $\qT_k$, and $\bar{\qH}_k$ are normalized such that $\tr\left(\qR_k\right)=N$, $\tr\left(\qT_k\right)=n_k$, and
$\tr\left(\bar{\qH}_k\bar{\qH}_k^H\right)=N$. In so doing, $P_k$ can be used as an indicator for the SNR of user $k$, which is independent from the matrix
dimensions. For notational brevity, henceforth, we assume that $P_k = 1~\forall k$ without loss of generality.\footnote{For practical applications, one can set any
positive value of $P_k$ without incurring any issues in the results of this paper. }

A diagonal structure in $\qR_k$ is sufficient to model the SCN scenarios under investigation. However, the more general nonnegative definite structure for $\qR_k$
could extend our results to cope with more complex applications. Due to the presence of $\bar\qH_k$, unfortunately, the required analysis is incredibly arduous. As
a result, if $\{ \bar\qH_k \neq \qzero\}_{\forall k}$, we restrict our consideration to diagonal $\qR_k$'s. Nevertheless, if $\{ \bar\qH_k = \qzero\}_{\forall k}$,
our results to be presented in Section III are valid even under nonnegative definite $\qR_k$'s.

\subsection*{B. Mutual Information and Stieltjes Transform}
Mutual information measures the achievable rate of a channel and has been a key metric for performance analysis in wireless communications. The Stieltjes transform
provides a convenient tool to study behavior of random matrices in large dimensional RMT. To do so, we first explain their relations.

Defining $\qH\triangleq\left[ \qH_1 \cdots \qH_K \right]$, $\bar{\qH}\triangleq\left[ \bar{\qH}_1 \cdots \bar{\qH}_K \right]$ and $n\triangleq\sum_{k=1}^{K} n_k$,
the mutual information of the MIMO channel can be linked to the eigenvalues of a nonnegative definite matrix $\qB_N$ of the form
\begin{equation}\label{eq:main_model}
\qB_N = \qS + \qH \qH^H=\qS + \sum_{k=1}^{K} \left(\qR_k^\frac{1}{2} \qX_k \qT_k^\frac{1}{2} + \bar{\qH}_k\right)\left(\qR_k^\frac{1}{2} \qX_k \qT_k^\frac{1}{2} +\bar{\qH}_k\right)^H,
\end{equation}
in which $\qS$ accounts for a source of correlated interference whose covariance matrix has the nonnegative square root $\qS^{\frac{1}{2}}$. Let $F_{\qB_N}$ be the
empirical spectral distribution (ESD) of the eigenvalues of $\qB_N$, given by
\begin{equation}
F_{\qB_N}(\lambda) = \frac{1}{N} \left\{\mbox{numbers of eigenvalues of $\qB_N \leq \lambda$} \right\}.
\end{equation}
One of the main problems in large dimensional RMT is to study the limiting spectral distribution (LSD) of $\qB_N$, denoted by $F_N$. A convenient tool for this is
the Stieltjes transform of $F_{\qB_N}(\lambda)$ which is defined as
\begin{equation}
m_{\qB_N}(z)\triangleq\int_{\bbR+} \frac{1}{\lambda-z} d F_{\qB_N}(\lambda) = \frac{1}{N} {\sf tr} \left( \qB_N - z \qI_{N} \right)^{-1}~~\mbox{for }z \in \bbC-\bbR^+.
\end{equation}
We will denote $\bbS(\bbR^+)$ as the class of all Stieltjes transforms of finite positive measures carried by $\bbR^+$. The Stieltjes transform provides a direct
way to identify the LSD of large dimensional random matrices. Some useful properties of the Stieltjes transforms are listed in Lemma \ref{StjLemma1}. According to
\cite{Mar-67,Sil-95}, to show that the difference between $F_{\qB_N}$ and $F_N$ converges vaguely to zero, it is equivalent to show that
\begin{equation} 
m_{\qB_N}(z) - m_N(z) \xrightarrow{a.s.} 0~~\mbox{for }z \in \bbC-\bbR^+,
\end{equation}
where $m_N(z) \triangleq \int_{\bbR+} \frac{1}{\lambda-z} d F_N(\lambda)$ is the Stieltjes transform of $F_N$.

For wireless communications, the importance of the Stieltjes transform is due to the fact that many important performance metrics can be expressed as functions of
the Stieltjes transform of $\qB_N$. The mutual information can be expressed as functionals of the Stieltjes transform of $\qB_N$ through the so-called Shannon
transform, where their relationship can be expressed as \cite[Section 2.2.3]{Tul-04} (or \cite[page 891]{Hac-07})
\begin{align}
\calV_{\qB_N}(\sigma^2) &\equiv \frac{1}{N} \log\det\left(\qI_N + \frac{1}{\sigma^2} \qH\qH^H \right) \notag \\
&= \int_{0}^{\infty} \log\left(1+ \frac{1}{\sigma^2}\lambda \right) d F_{\qB_N}(\lambda)  \notag\\
&= \int_{\sigma^2}^{\infty} \left( \frac{1}{\omega} - m_{\qB_N}(-\omega) \right) d\omega ~~~~~\mbox{for}~~\sigma^2 \in \bbR^+,
\label{eq:shannonTrans}
\end{align}
where it is assumed that $\qS = \qzero_N$ for simplicity.\footnote{The generalization of the corresponding results to the case with $\qS \neq \qzero$ is
straightforward. In the case with $\qS \neq \qzero$, the related performance metric (or the mutual information) is given by
\[
\frac{1}{N} \log\det\left(\qI_N + \qS + \frac{1}{\sigma^2} \qH\qH^H \right) - \frac{1}{N} \log\det\left(\qI_N + \qS \right).
\]} Here, $\calV_{\qB_N}(\sigma^2)$ provides a performance metric regarding the number of bits per second per Hertz per antenna that can be transmitted reliably over the SCN with channel matrices $\{\qH_k\}_{k =1,\dots,K}$.

In this paper, we are particularly interested in understanding the Stieltjes transform as well as the Shannon transform of $\qB_N$ in the asymptotic regime where
$K$ is fixed and $N, n_1, \dots, n_K$ all grow to infinity with ratios $\{ \beta_k(N) \equiv \frac{N}{n_k}\}_{k=1,\dots,K}$ such that
\begin{equation}
\beta_{\min} <\min_k\liminf_{N} \beta_k(N) < \max_k\limsup_{N} \beta_k(N) < \beta_{\max}
\end{equation}
and $0 < \beta_{\min}, \, \beta_{\max} < \infty$. For convenience, we will refer to this asymptotic regime simply as $\largeN\rightarrow \infty$. Our main goal is
to find a {\em nonrandom} matrix-valued function $\qPsi(z)$ (to be determined later) such that
\begin{equation} \label{eq:AidstjCong}
m_{\qB_N}(z) - \frac{1}{N} \tr(\qPsi(z))\xrightarrow{a.s.} 0~~\mbox{for }z \in \bbC-\bbR^+.
\end{equation}
This type of relation is referred to as {\em deterministic equivalent} \cite{Hac-07}, and $\frac{1}{N} \tr(\qPsi(z))$ is said to be the deterministic equivalent to
$m_{\qB_N}(z)$. We will apply (\ref{eq:AidstjCong}) to find a deterministic equivalent of the ergodic mutual information $\Ex \{\calV_{\qB_N}(\sigma^2)\}$, denoted
by $\calV_N(\sigma^2)$, and achieve this by proving $\Ex \{\calV_{\qB_N}(\sigma^2)\} - \calV_N(\sigma^2) \rightarrow 0$. In general, the computation of $\Ex
\{\calV_{\qB_N}(\sigma^2)\}$ relies on time-consuming Monte-Carlo computer simulations, while the deterministic equivalent is analytical and a lot easier to
compute than $\Ex \{\calV_{\qB_N}(\sigma^2)\}$.

In wireless communications applications, the Stieltjes transform itself may also be used to characterize the asymptotic signal-to-interference plus noise ratio
(SINR) of certain communication models, such as \cite{Tse-99}. The above-mentioned illustrations are just a few of the several important applications of the
Stieltjes transform in wireless communications. For a thorough survey of other applications, see \cite{Tul-04,Cou-11}. Undoubtedly, in order to construct reliable
applications in wireless communications, new analytical results concerning the LSD as well as the Stieltjes transform in the asymptotic regime are required.

\section*{\sc III. Main Results}
Before we present our main results, we first state the assumptions imposed in our SCN model.

\subsection*{A. Assumptions}
\begin{Assumption} \label{Ass1}
Let $\qX_k = [\frac{1}{\sqrt{n_k}} X_{ij}^{(k)}] \in \mathbb{C}^{N\times n_k}$, where $X_{ij}^{(k)}$'s are i.i.d. complex random variables with independent real
and imaginary parts such that $\Ex\{ X_{11}^{(k)}\}= 0$ and $\Ex \left\{|X_{11}^{(k)}-\Ex \{X_{11}^{(k)}\}|^2 \right\}=1$.
\end{Assumption}

\begin{Assumption} \label{Ass2}
The family of deterministic matrices $\{\qR_k,\qT_k,\qS\}_{\forall k}$ is deterministic nonnegative definite.
\end{Assumption}

\begin{Assumption} \label{Ass3}
The matrices $\qR_k$, $\qT_k$, and $\bar{\qH}_k$ are normalized such that
\begin{equation}\label{eq:normAss}
\left\{\begin{aligned}
\tr\left(\qR_k\right)&=N,\\
\tr\left(\qT_k\right)&=n_k,\\
\tr\left(\bar{\qH}_k\bar{\qH}_k^H\right)&=N.
\end{aligned}\right.
\end{equation}
\end{Assumption}

Clearly, because of the normalization constraint in (\ref{eq:normAss}), the sequences $\{F_{\qT_k}\}_{\forall k}$ are tight in $n_k$ while $\{F_{\qR_k}\}_{\forall
k}$ and $\{F_{\bar\qH_k\bar\qH_k^H}\}_{\forall k}$ are tight in $N$. It means that for each fixed $\epsilon \in (0,1)$, we can always select an $\alpha > 0$ such
that for all $n_k$, $F_{\qT_k}(\alpha) > 1 - \epsilon$, and for all $N$, $F_{\qR_k}(\alpha) > 1 - \epsilon$ and $F_{\bar\qH_k\bar\qH_k^H}(\alpha) > 1 - \epsilon$.

\begin{Assumption} \label{Ass5}
The family of deterministic matrices $\{\qR_k\}_{\forall k}$ is diagonal with nonnegative elements.
\end{Assumption}

Notice that Assumption \ref{Ass5} requires $\qR_k$ to be diagonal which is more restrictive than Assumption \ref{Ass2}. However, this assumption is still satisfied
in the application of the SCNs under investigation. Also, it should be noted that Assumption \ref{Ass5} is not required for some theorems presented in this paper.

\subsection*{B. Main Results}
We first introduce some properties of the deterministic matrix-valued function $\qPsi(z)$ which is needed in the deterministic equivalents of the Stieltjes
transform and the ergodic mutual information. To facilitate our expressions, we define the notation $\langle \qA \rangle_k$ that returns the submatrix of $\qA$
obtained by extracting the elements of the rows and columns with indices from $\sum_{i=1}^{k-1} n_i+1$ to $\sum_{i=1}^{k} n_i$.
\begin{Theorem} \label{mainTh_uniq}
Let $\beta_k=\frac{N}{n_k}$. Under Assumption \ref{Ass2}, the deterministic system of the following $K$ equations
\begin{subequations} \label{eq:fixedPoint}
\begin{align}
    e_i(z) &= \frac{1}{N} \tr\left( \qR_i \qPsi(z)\right) ~~~~~~~~\mbox{for }1 \leq i \leq K,\\
    \tilde{e}_i(z) &= \frac{1}{n_i} \tr\left( \qT_i \langle \tilde{\qPsi}(z) \rangle_{i}\right) ~~~~\mbox{for }1 \leq i \leq K,
\end{align}
\end{subequations}
where
\begin{subequations} \label{eq:PsiS}
\begin{align}
    \qPsi(z) &= \left( \qPhi(z)^{-1} - z \bar{\qH} \tilde{\qPhi}(z) \bar{\qH}^H \right)^{-1}, \label{eq:Psi} \\
    \tilde{\qPsi}(z) &= \left( \tilde{\qPhi}(z)^{-1} - z \bar{\qH}^H \qPhi(z) \bar{\qH} \right)^{-1}, \label{eq:tPsi} \\
    \qPhi(z) &= \frac{-1}{z}\left( -\frac{1}{z}\qS + \sum_{i=1}^{K} \tilde{e}_i(z)\qR_i + \qI_N \right)^{-1}, \label{eq:Phi} \\
    \tilde{\qPhi}(z)         &= \frac{-1}{z}\diag\left( (\qI_{n_1} + \beta_1 e_1(z)\qT_1)^{-1},\dots,(\qI_{n_K} + \beta_K e_K(z)\qT_K)^{-1} \right)\label{eq:tPhi}
\end{align}
\end{subequations}
have a unique solution for $z \in \bbC-\bbR^+$. In particular, $e_i(z)\in \bbS(\bbR^+)$ and $\tilde{e}_i(z) \in \bbS(\bbR^+)$ for $i \in \{1,\dots,K\}$.
\end{Theorem}

\begin{proof}
See Appendix B.
\end{proof}

We next provide a deterministic equivalent for the Stieltjes transform of $\qB_N$.

\begin{Theorem} \label{mainTh_Stj}
In addition to Assumptions \ref{Ass1}, \ref{Ass2}, and \ref{Ass3}, if one of the following conditions holds:
\begin{itemize}
\item[1)] $K=1$,
\item[2)] $\bar\qH = \qzero$ with $1 \leq K < \infty$,
\item[3)] Assumption \ref{Ass5} with $1 \leq K < \infty$,
\end{itemize}
then, as $\largeN\rightarrow \infty$, we have
\begin{equation} \label{eq:stjCong}
    m_{\qB_N}(z) - \frac{1}{N} \tr(\qPsi(z))\xrightarrow{a.s.} 0~~\mbox{for } z \in \bbC-\bbR^+.
\end{equation}

\end{Theorem}

\begin{proof}
Section III-C is dedicated to the proof of Theorem \ref{mainTh_Stj}.
\end{proof}

\begin{Remark}
According to (\ref{eq:Spatial_Cov}), we have addressed the non-central part of the channel through $\bar\qH_k$. Hence, we have $\Ex \{X_{11}^{(k)}\} = 0$ in
Assumption \ref{Ass1} for conceptual clarity. In fact, the assumption $\Ex\{ X_{11}^{(k)}\}= 0$ can be removed from Theorem \ref{mainTh_Stj} if $X_{ij}^{(k)}$'s
have the same mean. One can see that removing the same mean of the entries of $\qX_k$ does not affect the LSD of $F_{\qB_N}(x)$. See Appendix A.1 for detail.
\end{Remark}

When $K=1$, $\qS= \qzero$, and $\bar\qH_1 = \qzero$, Zhang's result in \cite{Zha-06} allows $\qR_1$ to be an arbitrary nonnegative definite matrix and $\qT_1$ to
be Hermitian. While $K=1$ and $\bar\qH_1 = \qzero$, Pan in \cite{Pan-10} tackled the cases where $\qT_1$ and $\qS$ are arbitrary Hermitian matrices, but $\qR_1
=\qI_N$. Using a new approach based on the generalized Lindeberg principle \cite{Kor-11}, we can handle both the cases in \cite{Zha-06,Pan-10} and make the proofs
simpler. As a result, Theorem \ref{mainTh_Stj} says that (\ref{eq:stjCong}) holds  when $K=1$ {\em without} the diagonal restriction on $\qR_1,\qT_1$, and {\em
without} the requirement of $\bar\qH_1 = \qzero$. This result also embraces the case in \cite[Section 3.2]{Hac-07} as a special case.

For the general case with $1 \leq K < \infty$ but $\bar\qH = \qzero$, it was pointed out by Couillet {\em et al.}~in \cite[Corollary 1]{Cou-09} that
(\ref{eq:stjCong}) holds when the entries of $\qX_k$ are i.i.d.~{\em Gaussian} random variables. When the entries of $\qX_k$ are non-Gaussain, (\ref{eq:stjCong})
was derived in \cite[Theorem 1]{Cou-09} under the assumption that $\qT_k$'s are diagonal. Therefore, Theorem \ref{mainTh_Stj} is more general than these previous
studies in the sense that the entries of $\qX_k$ are not necessarily Gaussian and $\qT_k$ is not necessarily diagonal, as conjectured in \cite{Cou-09}.

When $1 \leq K < \infty$ and the LOS component $\bar\qH$ is present, we require $\qR_k$'s to be diagonal due to mathematical difficulties. However, in this case it
is worth pointing out that (\ref{eq:stjCong}) is also true if the matrices $\qR_1, \dots, \qR_K$ are simultaneously unitary diagonalizable according to our
argument. This type of channel model is the so-called ``virtual channel representation'' in \cite{Say-02} and is found to be useful for modeling channels with many
antennas \cite{Ozc-05}.

As an application, we next use Theorem \ref{mainTh_Stj} to provide a deterministic equivalent of the ergodic mutual information in the following theorem.

\begin{Theorem} \label{mainTh_Cap}
Assume that $\qB_N$ follows the hypotheses of Theorem \ref{mainTh_Stj} and $\qS = \qzero$. Then, as $\largeN\rightarrow \infty$, the Shannon transform of $\qB_N$
satisfies
\begin{equation} \label{eq:a4}
    \Ex\{ \calV_{\qB_N}(\sigma^2) \} - \calV_N(\sigma^2) \longrightarrow 0,
\end{equation}
where
\begin{multline}\label{eq:AsyShannon}
\calV_N(\sigma^2) = \frac{1}{N} \log\det\left( \frac{\qPhi(-\sigma^2)^{-1}}{\sigma^2} + \bar{\qH} \tilde{\qPhi}(-\sigma^2) \bar{\qH}^H \right) +
\frac{1}{N} \log\det\left( \frac{\tilde{\qPhi}(-\sigma^2)^{-1}}{\sigma^2} \right)\\
-\sigma^2 \sum_{i=1}^{K} e_i(-\sigma^2) \tilde{e}_i(-\sigma^2).
\end{multline}
\end{Theorem}

\begin{proof}
(\ref{eq:AsyShannon}) is an explicit expression of $\int_{\sigma^2}^{\infty} \left( \frac{1}{\omega} - \frac{1}{N} \tr\left(\qPsi(-\omega)\right) \right) d\omega$.
The proofs of the convergence and the explicit expression are given in Appendix C.
\end{proof}

For the case with $K = 1$ and $\bar\qH \neq \qzero$ and the general case with $1 \leq K < \infty$ but $\bar\qH = \qzero$, $\calV_N(\sigma^2)$ agrees perfectly with
those in \cite[Theorem 1]{Dum-10} and \cite[Theorem 2]{Cou-09}, respectively. Nevertheless, Theorem \ref{mainTh_Cap} is more general than \cite[Theorem 1]{Dum-10}
and \cite[Theorem 2]{Cou-09} in the sense that there is no Gaussian distribution requirement on the entries of $\qX_k$. Note that in the above two cases, Theorem
\ref{mainTh_Cap} allows $\qR_k$'s and $\qT_k$'s to be generally nonnegative definite. Further, for the general case with $1 \leq K < \infty$ and $\bar\qH \neq
\qzero$, Theorem \ref{mainTh_Cap} contains \cite[Theorem 4.1]{Hac-07} as a special case even though Theorem \ref{mainTh_Cap} requires $\qR_k$'s to be diagonal,
whereas in \cite[Theorem 4.1]{Hac-07}, both $\qR_k$'s and $\qT_k$'s are restricted to be diagonal. Finally, unlike several of other contributions (e.g.,
\cite[Theorem 1]{Dum-10} and \cite[Theorem 4.1]{Hac-07}), where $\qR_k$'s, $\qT_k$'s, and $\bar{\qH}_k$'s are required to have uniformly bounded spectral norms,
Theorem \ref{mainTh_Cap} is valid for the more general trace constraints (\ref{eq:normAss}). This relaxation makes Theorem \ref{mainTh_Cap} valid for all possible
correlation patterns and LOS components.

As the large-system results are invariant to the type of fading distribution, any designs based on the large-system results are robust, and the properties of the
asymptotic optimal input covariance are invariant to the type of fading distribution. Specifically, by \cite[Proposition 3]{Cou-09}, we conclude that if $\bar\qH
=\qzero$, even when the entries of $\qX_k$ are non-Gaussian, the eigenvectors of the asymptotic optimal input covariance matrix align with that of $\qT_k$ while
the eigenvalues follows a water-filling principle. In \cite{Dum-10,Cou-09}, an iterative water-filling algorithm based on $\calV_N(\sigma^2)$ is provided to obtain
the asymptotic optimal input covariance. The iterative algorithm turns out to have wide applicability to all types of fading distribution.

Unlike \cite[Theorem 2]{Cou-09}, Theorem \ref{mainTh_Cap} does not assert that $\calV_{\qB_N}(\sigma^2) - \calV_N(\sigma^2) \xrightarrow{a.s.} 0$. Although
(\ref{eq:a4}) has already satisfied our applications of interest, we find it important to clarify some properties regarding the a.s.~convergence. Indeed, following
\cite[Theorem 2]{Cou-09} and using Theorem \ref{mainTh_Stj}, (\ref{eq:a4}) can be strengthened to a.s.~convergence under an additional assumption stated in the
following theorem.
\begin{Theorem} \label{AsConv_Cap}
In addition to the assumptions of Theorem \ref{mainTh_Cap}, suppose further that
\begin{enumerate}
\item[1)] $\Ex\left\{ |X_{11}^{(k)}|^4 \right\} < \infty$;

\item[2)]
There exists an $\alpha$ and a sequence $\tau_N$ such that for all $N$,
\[
    \max_{k}\,\max\left\{ \lambda_{\tau_N+1}(\qR_k), \lambda_{\tau_N+1}(\qT_k), \lambda_{\tau_N+1}(\bar{\qH}_k\bar{\qH}_k^H) \right\} \leq \alpha,
\]
where $\lambda_i(\qA)$ denotes the $i$th largest eigenvalue of a matrix $\qA$.

\item[3)]
Let $b_N$ denote an upper-bound on the spectral norm of $\{\qT_k, \qR_k, \bar{\qH}_k\bar{\qH}_k^H\}_{1 \leq k \leq K}$, and $c > 0$ a constant such that $c >
\frac{ K \beta_{\max}}{\beta_{\min}} \left(1+\sqrt{\beta_{\min}}\right)^2$, $a_N = c b_N^2 $ satisfies $$\tau_N
\log{\left(1+\frac{a_N}{\sigma^2}\right)}=o(N).$$
\end{enumerate}
Then, (\ref{eq:a4}) can be strengthened as
\begin{equation} \label{eq:a5}
\calV_{\qB_N}(\sigma^2) - \calV_N(\sigma^2) \xrightarrow{a.s.} 0.
\end{equation}
\end{Theorem}

\begin{Remark}
Indeed, (\ref{eq:a5}) holds if assumptions 2) and 3) of Theorem \ref{AsConv_Cap} are replaced by the assumptions that $\frac{1}{N}\tr(\qR_k^4)$,
$\frac{1}{n_k}\tr(\qT_k^4)$, and $\frac{1}{N}\tr\left((\bar{\qH}_k\bar{\qH}_k^H)^2\right)$ are bounded for all $k$. A proof is given in Appendix C.2.
\end{Remark}

Note that neither of the assumptions in Theorem \ref{AsConv_Cap} and Remark 2 implies the other. It was pointed out in \cite{Cou-09} that most conventional models
for $\qR_k$'s and $\qT_k$'s satisfy the assumptions of Theorem \ref{AsConv_Cap}. However, the assumptions in Remark 2 assist in covering some cases that are not
met by 2) and 3) of Theorem \ref{AsConv_Cap}.\footnote{As an example, suppose that the eigenvalues of $\qR_k$ are given as: one being $N^\frac{1}{4}$,
$\frac{N}{\log N}$ of them being $(\log N)^\frac{1}{4}$, and the remaining eigenvalues being bounded. In this case, $\tau_N = \frac{N}{\log N}$ and $a_N =
O(N^\frac{1}{2})$, so $\tau_N \log(1+\frac{a_N}{\sigma^2}) = O(N)$ rather than $o(N)$. Therefore, assumptions 2) and 3) of Theorem \ref{AsConv_Cap} are not
satisfied. However, $\frac{1}{N}\tr(\qR_k^4)$ is bounded.}

\subsection*{C. Proof of Theorem \ref{mainTh_Stj}}
This subsection gives an outline of the proof of Theorem \ref{mainTh_Stj} for ease of understanding.

As in \cite[Section 4.5.1]{Bai-10}, we argue that the entries of $\qX_k$ can be replaced by random variables bounded in absolute value by
$\varepsilon_{n}\sqrt{n_k}$ without changing the LSD of $F_{\qB_N}$. Here $\varepsilon_n$ is a positive sequence converging to zero. Also, following \cite[Section
4.3.1]{Bai-10} we apply a truncation on $\qR_k, \qT_k, \bar\qH_k\bar\qH_k^H$ such that the spectral norms of $\qR_k$, $\qT_k$, and $\bar{\qH}_k\bar{\qH}_k^H$ are
bounded by a constant, say $\alpha$ (see details in Appendix A.1). As a consequence, the proof of Theorem \ref{mainTh_Stj} may be achieved under the following
additional conditions.
\begin{Assumption}\label{Ass4}
For each of $N$, $X_{ij}^{(k)}$ are i.i.d., and
\begin{subequations} \label{eq:asumXH}
\begin{align}
 & \Ex \left\{X_{11}^{(k)}\right\}=0, ~ \Ex\left\{|X_{11}^{(k)}|^2\right\}=1,~|X_{11}^{(k)}| \leq \varepsilon_{n}\sqrt{n_k},\\
 & \max_{k=1,\dots,K}\,\max\{\|\qR_k\|,\|\qT_k\|,\|\bar{\qH}_k\bar{\qH}_k^H\|\} \leq \alpha.
\end{align}
\end{subequations}
\end{Assumption}
For convenience, we still use $\qX_k$, $\qR_k$, $\qT_k$, and $\bar{\qH}_k$ to denote those truncated and centralized matrices.

Write
\begin{multline}\label{eq:threeStep}
m_{\qB_N}(z)-\frac{1}{N} \tr\left(\qPsi(z)\right)\\
=\left(m_{\qB_N}(z) - \Ex \{m_{\qB_N}(z)\} \right)+ \left(\Ex \{ m_{\qB_N}(z) \} - \Ex \{ m_{\calqB_N}(z) \} \right)+\left(\Ex\{ m_{\calqB_N}(z) \} - \frac{1}{N}
\tr\left(\qPsi(z)\right)\right),
\end{multline}
where $\calqB_N$ is obtained from $\qB_N$ defined in (\ref{eq:main_model}) with all the entries $X_{ij}^{(k)}$'s of $\qX_k$ replaced by independent standard {\em
Gaussian} random variables $\calX_{ij}^{(k)}$'s. Here $X_{ij}^{(k)}$'s are independent of $\calX_{ij}^{(k)}$'s.

The proof of Theorem \ref{mainTh_Stj} then consists of the following three steps:
\begin{itemize}
\item[{\bf Step 1.}]
By a martingale approach we first prove that
\begin{equation} \label{eq:m-em}
    m_{\qB_N}(z) - \Ex \{m_{\qB_N}(z)\} \buildrel{a.s.}\over{\longrightarrow} 0.
\end{equation}

\item[{\bf Step 2.}]
By the Lindeberg principle (Lemma \ref{Lindeberg} below) \cite[Theorem 2]{Kor-11} we claim that
\begin{equation}\label{eq:GAndNoGDiff}
    \Ex \{ m_{\qB_N}(z) \} - \Ex \{ m_{\calqB_N}(z) \} \longrightarrow 0.
\end{equation}

\item[{\bf Step 3.}]
The last step is to investigate the Stieltjes transform of $\calqB_N$ so that the following is true:
\begin{equation}\label{eq:ExMsM}
    \Ex\{ m_{\calqB_N}(z) \} - \frac{1}{N} \tr \qPsi(z) \longrightarrow 0.
\end{equation}
\end{itemize}

Not that Step 1 and Step 2 are completed under Assumptions \ref{Ass1}--\ref{Ass3} in which $\qR_k$'s and $\qT_k$'s are generally nonnegative definite and $\bar\qH
\neq \qzero$. Appendix A.2 handles Step 1 while Step 2 is addressed in Appendix A.3, where we mainly make use of the generalized Lindeberg principle given below.
\begin{Lemma}\label{Lindeberg}
{\rm (Generalized Lindeberg Principle \cite{Kor-11})} Let $\qv = [v_i] \in \bbR^{n}$ and $\tilde\qv = [\tilde{v}_i] \in \bbR^{n}$ be two random vectors with
mutually independent components. Define $\{a_i\}_{1\leq i \leq n}$ and $\{b_i\}_{1\leq i \leq n}$ with
\begin{equation} \label{eq:LindeAB}
a_i = |\Ex \{v_i\} - \Ex \{\tilde{v}_i\}|,~~\mbox{and}~~b_i = |\Ex \{v_i^2\} - \Ex \{\tilde{v}_i^2\}|.
\end{equation}
Then, given a thrice continuously differentiable function $f: \bbR^n \rightarrow \bbR$, we have
\begin{multline}\label{eq:LindeCheck}
|\Ex \{f(\qv)\} - \Ex \{f(\tilde\qv)\}|\leq\sum_{i=1}^{n}\left[a_i \, \Ex \{|\partial_i f(\qv_1^{i-1},0,\tilde{\qv}_{i+1}^{n})|\} + \frac{1}{2} b_i \, \Ex \{|\partial_i^2 f(\qv_1^{i-1},0,\tilde{\qv}_{i+1}^{n})|\}\right.\\
+\frac{1}{2} \Ex \left\{ \int_{0}^{v_i} |\partial_i^3 f(\qv_1^{i-1},s,\tilde{\qv}_{i+1}^{n})| (v_i-s)^2 ds \right\}\\
\left.+\frac{1}{2} \Ex \left\{ \int_{0}^{\tilde{v}_i} |\partial_i^3 f(\qv_1^{i-1},s,\tilde{\qv}_{i+1}^{n})| (\tilde{v}_i-s)^2 ds \right\}\right],
\end{multline}
where $\partial_i^p$ is the $p$-fold derivative in the $i$th coordinate, $\qv_1^{i-1} = (v_1,\dots, v_{i-1})$, and $\tilde{\qv}_{i+1}^{n}
=(\tilde{v}_{i+1},\dots,\tilde{v}_{n})$.
\end{Lemma}

Here, it should be noted that both Lindeberg's principle (Chatterjee \cite{Cha-06}, Korada and Montanari \cite{Kor-11}) and the interpolation trick (see Lytova and
Pastur in \cite{Lyt-09} and Pan \cite{Pan-10}) can be used to handle Step 2. However, Lindeberg's principle is simpler than the interpolation trick when proving
this type of problems.

Due to Lemma \ref{Lindeberg}, the remaining task is to consider $\calqB_N$ with underlying random variables being a standard Gaussian distribution. In the
remainder of this subsection, we will show how Step 3 can be done case by case from the known results of Gaussian matrices.

Before proceeding, we recall some useful results. Denote the spectral decomposition of $\qT_k^\frac{1}{2}$ by $\tilde\qU_k \tilde\qD_k^\frac{1}{2}\tilde\qU_k^H$,
where $\tilde\qU_k \in\bbC^{n_k \times n_k}$ is unitary and $\qD_k$ is diagonal. Since $\calqX_k$ is Gaussian, the joint distribution of $\calqX_k \tilde\qU_k$ is
the same as that of $\calqX_k$. Thus, $(\qR_k^\frac{1}{2} \calqX_k \qT_k^\frac{1}{2} +\bar{\qH}_k)(\qR_k^\frac{1}{2} \calqX_k \qT_k^\frac{1}{2} + \bar{\qH}_k)^H$
has the same distribution as
\begin{equation}\label{eq:RXD}
\left(\qR_k^\frac{1}{2} \calqX_k \tilde\qD_k^\frac{1}{2} + \bar{\bar\qH}_{k}\right)\left(\qR_k^\frac{1}{2} \calqX_k \tilde\qD_k^\frac{1}{2} + \bar{\bar\qH}_{k}\right)^H,
\end{equation}
where $\bar{\bar\qH}_{k}\triangleq\bar{\qH}_k\tilde\qU_k$. Hence, we will assume that the channel is in the form of
$\qR_k^\frac{1}{2}\calqX_k\tilde\qD_k^\frac{1}{2} + \bar{\bar\qH}_{k}$ in the sequel.

Consider condition 1) of Theorem \ref{mainTh_Stj} (i.e., the case $K=1$) first. Denote the spectral decomposition of $\qR_1^\frac{1}{2}$ by $\qU_1
\qD_1^\frac{1}{2}\qU_1^H$. For simplicity, we assume $\qS = \qzero$.\footnote{The extension to the case with $\qS\ne{\bf 0}$ is straightforward.} The joint
distribution of $\qU_k^H \calqX_k $ is the same as that of $\calqX_k $. Therefore, the distribution of $(\qR_1^\frac{1}{2}\calqX_1 \tilde\qD_1^\frac{1}{2} +
\bar{\bar\qH}_{1})(\qR_1^\frac{1}{2} \calqX_1 \tilde\qD_1^\frac{1}{2} + \bar{\bar\qH}_{1})^H$ is the same as that of $(\qU_1\qD_1^\frac{1}{2}
\calqX_1\tilde\qD_1^\frac{1}{2} + \bar{\bar\qH}_{1})(\qU_1\qD_1^\frac{1}{2} \calqX_1 \tilde\qD_1^\frac{1}{2} + \bar{\bar\qH}_{1})^H$. Since
\begin{multline}
\tr\left( \left( (\qU_1\qD_1^\frac{1}{2}\calqX_1\tilde\qD_1^\frac{1}{2} + \bar{\bar\qH}_{1})(\qU_1\qD_1^\frac{1}{2} \calqX_1 \tilde\qD_1^\frac{1}{2} +
\bar{\bar\qH}_{1})^H -z \qI_N \right)^{-1} \right)
 \\ =\tr\left( \left( (\qD_1^\frac{1}{2} \calqX_1\tilde\qD_1^\frac{1}{2} + \bar{\bar{\bar\qH}}_{1})(\qD_1^\frac{1}{2} \calqX_1 \tilde\qD_1^\frac{1}{2} + \bar{\bar{\bar\qH}}_{1})^H -z \qI_N \right)^{-1} \right)
\end{multline}
with $\bar{\bar{\bar\qH}}_{1} \triangleq\qU_1^H\bar{\qH}_1\tilde\qU_1$, it suffices to prove (\ref{eq:ExMsM}) with $\calqB_N$ as follows
\begin{equation} \label{eq:BN_GaussianSU}
    \calqB_N = \left(\qD_1^\frac{1}{2} \calqX_1\tilde\qD_1^\frac{1}{2} + \bar{\bar{\bar\qH}}_{1}\right)\left(\qD_1^\frac{1}{2} \calqX_1 \tilde\qD_1^\frac{1}{2} + \bar{\bar{\bar\qH}}_{1}\right)^H.
\end{equation}
However, the convergence of (\ref{eq:ExMsM}) with $\calqB_N$ in (\ref{eq:BN_GaussianSU}) was reported in \cite[Section 3.2]{Hac-07}. Indeed, in \cite{Hac-07},
$X_{ij}^{(1)}$'s are i.i.d.~complex random variables with finite fourth moment and the variance profile is separable. For our interest, $\calX_{ij}^{(1)}$'s are
standard Gaussian random variables and the variance profile of the channel $\qD_1^\frac{1}{2} \calqX_1\tilde\qD_1^\frac{1}{2}$ is separable. Hence, Step 3 is
completed under condition 1) of Theorem \ref{mainTh_Stj}.

Next, we turn to condition 2) of Theorem \ref{mainTh_Stj} (i.e., the case $\bar\qH = \qzero$). In view of (\ref{eq:RXD}), it is enough to prove (\ref{eq:ExMsM})
with $\calqB_N$ as
\begin{equation} \label{eq:BN_GaussianMUwoLOS}
\calqB_N = \qS + \sum_{k=1}^{K} \qR_k^\frac{1}{2} \calqX_k \tilde\qD_k \calqX_k^H \qR_k^\frac{1}{2}.
\end{equation}
The convergence of (\ref{eq:ExMsM}) with $\calqB_N$ in (\ref{eq:BN_GaussianMUwoLOS}) follows immediately from \cite[Corollary 1]{Cou-09}. Therefore, Step 3 is
finished under condition 2) of Theorem \ref{mainTh_Stj}.

Finally, we consider condition 3) of Theorem \ref{mainTh_Stj} (i.e., $\qR_k$'s being diagonal). Note that $\qR_k$'s are diagonal nonnegative matrices in the
remainder of this subsection. From (\ref{eq:RXD}), it suffices to prove (\ref{eq:ExMsM}) with $\calqB_N$ as
\begin{equation} \label{eq:BN_GaussianMUwLOS}
\calqB_N = \qS + \sum_{k=1}^{K} \left( \qR_k^\frac{1}{2} \calqX_k \tilde\qD_k^\frac{1}{2} + \bar{\bar\qH}_k \right)\left(\qR_k^\frac{1}{2} \calqX_k \tilde\qD_k^\frac{1}{2} + \bar{\bar\qH}_k \right)^H.
\end{equation}
The following theorem contributes to Step 3 in this case.

\begin{Theorem} \label{Theorem_Hac07}
Consider the channel matrix of the form $\qH_k=\qR_k^\frac{1}{2} \calqX_k \tilde\qD_k^\frac{1}{2} + \bar{\bar{\qH}}_k$ for $k=1,\dots,K$, where $\qR_k$'s and
$\tilde\qD_k$'s are diagonal nonnegative matrices. Assume that the spectral norms of $\qR_k$'s, $\tilde\qD_k$'s, and $\bar{\bar{\qH}}_k$'s are all bounded and that
$\calX_{ij}^{(k)}$'s are i.i.d.~standard Gaussian random variables. Let $\bar{\bar\qH}=\left[\bar{\bar\qH}_1\cdots \bar{\bar\qH}_K \right]$. As $\largeN\rightarrow
\infty$, then for any $z \in \bbC-\bbR^+$,
\begin{subequations} \label{eq:asyHH}
\begin{align}
\frac{1}{N} \Ex \left\{ \tr \left( \qH \qH^H -z\qI \right)^{-1} \right\} - \frac{1}{N} \tr( \qXi(z)) &\longrightarrow 0,\label{eq:asyQHH1}\\
\frac{1}{N} \Ex \left\{ \tr \left( \qH^H \qH -z\qI \right)^{-1} \right\} - \frac{1}{N} \tr( \tilde\qXi(z)) &\longrightarrow 0,\label{eq:asyHH2}
\end{align}
\end{subequations}
where
\begin{subequations}
\begin{align}
\qXi(z) &= \left( \qTheta(z)^{-1} - z \bar{\bar{\qH}} \tilde{\qTheta}(z) \bar{\bar{\qH}}^H \right)^{-1}, \label{eq:Xi}\\
\tilde{\qXi}(z) &= \left( \tilde{\qTheta}(z)^{-1} - z \bar{\bar{\qH}}^H \qTheta(z) \bar{\bar{\qH}} \right)^{-1}, \label{eq:tXi}\\
\qTheta(z) &= \frac{-1}{z}\left( \sum_{i=1}^{K}  n_i^{-1} \tr \left(\tilde\qD_i \ang{\tilde{\qXi}(z)}_{i} \right) \qR_i + \qI_N \right)^{-1}, \\
\tilde{\qTheta}(z) &= \frac{-1}{z}\diag \left( \left(\qI_{n_1} + n_1^{-1} \tr \left(\qR_1 \qXi(z)\right) \tilde\qD_1\right)^{-1},\dots, \left(\qI_{n_K} + n_K^{-1} \tr \left(\qR_K \qXi(z)\right) \tilde\qD_K\right)^{-1} \right). \label{eq:tTheta}
\end{align}
\end{subequations}
\end{Theorem}

\begin{proof}
Note that $\qH = \left[ \qR_1^\frac{1}{2} \calqX_1 \tilde\qD_1^\frac{1}{2}\cdots \qR_K^\frac{1}{2} \calqX_K \tilde\qD_K^\frac{1}{2} \right] + \bar{\bar{\qH}} $
corresponds to the channel in \cite{Hac-07} with a general variance profile. Therefore, this theorem can be obtained immediately from \cite[Theorems 2.4 and
2.5]{Hac-07} and the dominated convergence theorem.\footnote{The dominated convergence theorem is due to the expectation involved in (\ref{eq:asyHH}).}
\end{proof}

The proof of (\ref{eq:ExMsM}) under condition 3) of Theorem \ref{mainTh_Stj} is a result of Theorem \ref{Theorem_Hac07}. The idea is to cast the model $\qH\qH^H
+\qS$ into an extended model such that it fits into the framework of (\ref{eq:asyHH}). To this end, write
\begin{equation*}
\qH\qH^H +\qS = \sum_{k=1}^{K+1} \qH_k \qH_k^H,
\end{equation*}
where $\qH_k = \qR_k^\frac{1}{2} \calqX_k \tilde\qD_k^\frac{1}{2} + \bar{\bar\qH}_k$ is given in (\ref{eq:BN_GaussianMUwLOS}) and $\qH_{K+1} = \qS^\frac{1}{2}$
(without a random component). Plugging this model into (\ref{eq:Xi}) and (\ref{eq:tXi}), we obtain
\begin{align}
\qXi(z) &= \left( \qTheta(z)^{-1} - z \bar{\bar\qH} \tilde{\qTheta}(z) \bar{\bar\qH}^H + \qS \right)^{-1}, \\
\tilde{\qXi}(z) &= \left(\left[\begin{array}{cc}
                       \tilde{\qTheta}(z)^{-1} & \qzero \\
                       \qzero & -z \qI_N
                       \end{array}
                       \right]
                       - z
                       \left[\begin{array}{c}
                       \bar{\bar\qH}^H  \\
                       \qS^\frac{1}{2}
                       \end{array}
                       \right]
                       \qTheta(z)
                       \left[\begin{array}{cc}
                       \bar{\bar\qH} & \qS^\frac{1}{2}
                       \end{array}\right]\right)^{-1}.
\end{align}
Note that $\tilde{\qXi}(z)$ is now a matrix of size $(n+N) \times (n+N)$. From (\ref{eq:tTheta}), write
\begin{equation}\label{eq:tTheAndS}
\qTheta(z)^{-1} + \qS = - \sum_{i=1}^{K} \frac{z}{n_i}\tr\left( \tilde\qD_i \langle \tilde{\qXi}(z) \rangle_{i}\right) \qR_i -z\qI_N + \qS.
\end{equation}
Applying Lemma \ref{invLemma} in Appendix D, we can obtain the $n \times n$ principal submatrix of $\tilde{\qXi}(z)$ as
\begin{align}
[\tilde{\qXi}(z)]_{1:n,1:n}
= & \left[ \tilde{\qTheta}(z)^{-1} - z\bar{\bar\qH}^H \qTheta(z) \bar{\bar\qH} -
z^2 \bar{\bar\qH}^H \qTheta(z) \qS^\frac{1}{2} \left( -z\qI_N - z\qS^\frac{1}{2} \qTheta(z) \qS^\frac{1}{2} \right)^{-1} \qS^\frac{1}{2} \qTheta(z) \bar{\bar\qH}
    \right]^{-1} \notag \\
    =& \left[  \tilde{\qTheta}(z)^{-1} - z\bar{\bar\qH}^H \left( \qTheta(z) + z \qTheta(z) \qS^\frac{1}{2}
    \left( -z\qI_N - z\qS^\frac{1}{2} \qTheta(z) \qS^\frac{1}{2} \right)^{-1} \qS^\frac{1}{2} \qTheta(z)
    \right) \bar{\bar\qH}
    \right]^{-1} \notag \\
    =& \left[  \tilde{\qTheta}(z)^{-1} - z\bar{\bar\qH}^H
    \left( \qTheta(z)^{-1} + \qS \right)^{-1} \bar{\bar\qH}
    \right]^{-1}, \label{eq:subMat}
\end{align}
where the third equality is due to the matrix inverse lemma (Lemma \ref{invLemma} in Appendix D). Plugging (\ref{eq:tTheAndS}) and (\ref{eq:subMat}) into
(\ref{eq:Xi})--(\ref{eq:tTheta}) and recovering the effect from the eigenvectors $\qU_k$'s, we obtain the formulas (\ref{eq:Psi})--(\ref{eq:tPhi}). In particular,
defining $\tilde{\qU} = \diag(\tilde{\qU}_1,\dots,\tilde{\qU}_K)$, we have $\qPsi(z) = \qTheta(z)^{-1} + \qS$, $\tilde{\qPsi}(z) =\tilde{\qU}^H
[\tilde{\qXi}(z)]_{1:n,1:n} \tilde{\qU}$, $\qTheta(z) =\qPhi(z)$, and $\tilde\qTheta(z) = \tilde{\qU}^H\tilde\qPhi(z)\tilde{\qU}$. By (\ref{eq:asyQHH1}), we
immediately establish (\ref{eq:ExMsM}).

For the general case with nontrivial $\bar\qH$, $\qR_k$'s are required to be diagonal so that \cite[Theorem 2.4]{Hac-07} can be used immediately to yield Theorem
\ref{Theorem_Hac07}. If $\qR_k$'s are generally nonnegative definite, one may wonder if the same Stieltjes transform method in \cite{Hac-07} can still be used to
get the similar result. At present, due to mathematical difficulties, this is still an open challenge and such development is ongoing.

\section*{\sc IV. Simulation Results}
In this section, computer simulations are provided to evaluate the reliability of the asymptotic result particularly when the channel entries are non-Gaussian.
Specifically, we compare the analytical result $\calV_{N}(\sigma^2)$ (\ref{eq:AsyShannon}) with the Monte-Carlo simulation results of the ergodic mutual
information $\Ex\{\calV_{\qB_N}(\sigma^2)\}$ obtained from averaging over a large number of independent realizations of $\qH$.

Given the Kronecker MIMO channel model $\qH_k = \qR_k^\frac{1}{2} \qX_k \qT_k^\frac{1}{2}+\bar{\qH}_k$ for ${\sf UE}_k$, the simulation settings used in this study
are based on the following assumptions. First, the spatial correlation is generated from a uniform linear array with half wavelength spacing in a wireless
scenario. The propagation path cluster is assumed to have a Gaussian power azimuthal distribution, which is characterized by the mean angle and the
root-mean-square spread \cite{Mou-03}. Second, the channel gain from ${\sf UE}_k$ to each receiving antenna $\qR_k$ and its LOS components $\bar\qH_k$ are
generated randomly. Third, the i.i.d. entries of $\qX_k$'s are assumed to be of the form $\frac{1}{\sqrt{n_k}} X_{11}^{(k)} = \frac{1}{\sqrt{n_k}} W_{11}^{(k)}
\exp(j \theta_{11}^{(k)})$ \cite{Cho-07}, where $\theta_{11}^{(k)}$ is the phase modeled as a uniform distribution over $[0,2\pi]$, and $W_{11}^{(k)}$ is the
random amplitude drawn from a distribution with normalized mean power, i.e., $\Ex\{[W_{11}^{(k)}]^2\}=1$. The typical probability distributions for modeling the
amplitude behavior include the Rayleigh, Nakagami, and log-normal distributions \cite{Mol-05,Foe-03,Mol-05}. Among them, the Nakagami distribution is arguably the
most general model that embraces the Rayleigh distribution and those having longer tails. On the other hand, the log-normal distribution is well known to be a
suitable model for slowly varying communication channels, e.g., indoor radio propagation environments.

To measure the fading severity of the channel model, we adopt the coefficient of variance (CV) as a performance metric, which is defined by \cite{Cho-07}
\begin{equation}
{\rm CV}=\frac{\sqrt{{\sf var}\{W_{11}^{(k)}\}}}{\Ex\{W_{11}^{(k)}\}}
\end{equation}
with ${\sf var}\{W_{11}^{(k)}\}$ being the variance of $W_{11}^{(k)}$. According to \cite{Cho-07}, the variation in ergodic mutual information can be significant
if the values of CV are different. Note that the CV for Rayleigh fading channels is $0.526$ and any CV value much greater than this reference point indicates a
severe level of fading. For Nakagami fading, fading is severe if the Nakagami $m$-factor is very small. However, the $m$-factor is greater than $0.5$
\cite{Nak-60}, which gives a possible range for the CV values only in $[0,0.7555]$. Therefore, we use the log-normal distribution to generate a fading channel with
very severe fading by setting a large value for CV.

\begin{figure}
\begin{center}
\resizebox{7.5in}{!}{%
\includegraphics*{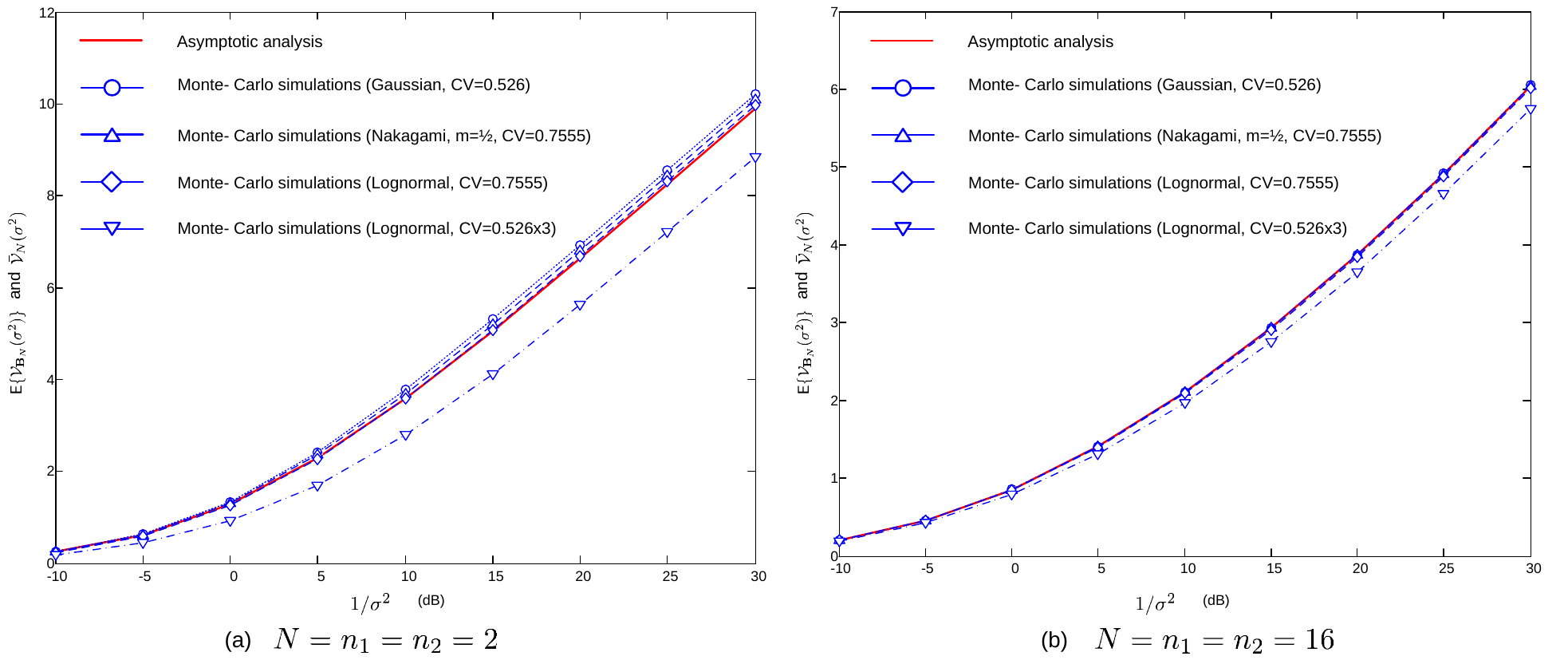} }%
\caption{Ergodic mutual information versus SNRs for the SCNs with $\bar\qH = \qzero$ and a) $N=n_1=n_2=2$ and b) $N=n_1=n_2=16$. The solid lines plot the analytical results, while the markers plot the exact results.}\label{fig:simAnAnalK0}
\end{center}
\end{figure}

\begin{figure}
\begin{center}
\resizebox{7.5in}{!}{%
\includegraphics*{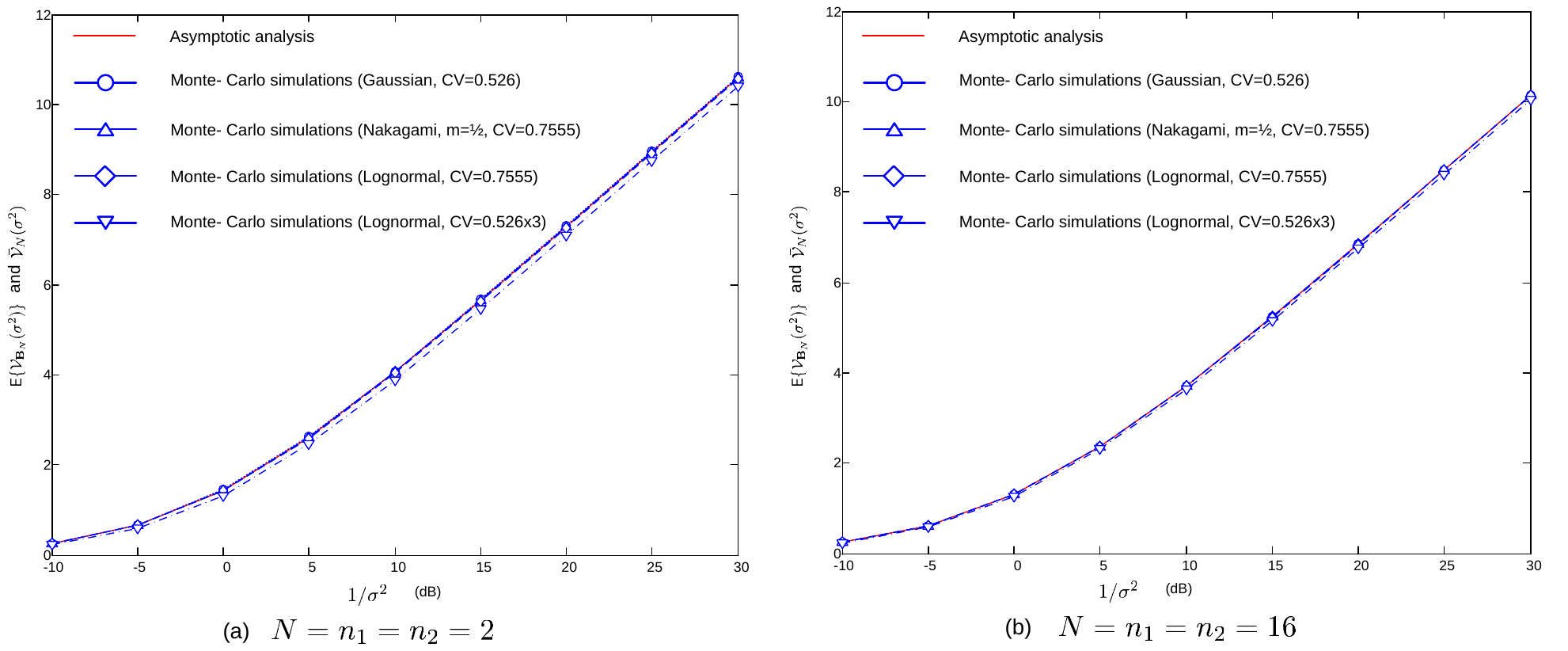} }%
\caption{Ergodic mutual information versus SNRs for the SCNs with $\bar\qH\ne\qzero$ and a) $N=n_1=n_2=2$ and b) $N=n_1=n_2=16$. The solid lines plot the analytical results, while the markers plot the exact results.}\label{fig:simAnAnalK1}
\end{center}
\end{figure}

Under a different fading severity, Figures \ref{fig:simAnAnalK0} and \ref{fig:simAnAnalK1} show the results of $\Ex\{\calV_{\qB_N}(\sigma^2)\}$ and
$\calV_N(\sigma^2)$ for the cases with $\bar\qH = \qzero$ and $\bar\qH \neq \qzero$ respectively. As we can see, when the number of antennas grows large (e.g.,
$N=n_1=n_2=16$) all curves almost overlap regardless of the distributions or the CV values. The ergodic mutual information is more sensitive to the type of
distribution as well as the CV value for the scenarios with small number of antennas. Thus, this invariance phenomenon of the ergodic mutual information in the
large-system limit agrees with our analysis. Also, one can observe that the case $\bar\qH \neq \qzero$ exhibits less sensitivity to the type of distribution, even
for a small number of antennas because half of the energy has contributed to the LOS components which has nothing to do with mitigating the fading distributions.

\begin{figure}
\begin{center}
\resizebox{4.5in}{!}{%
\includegraphics*{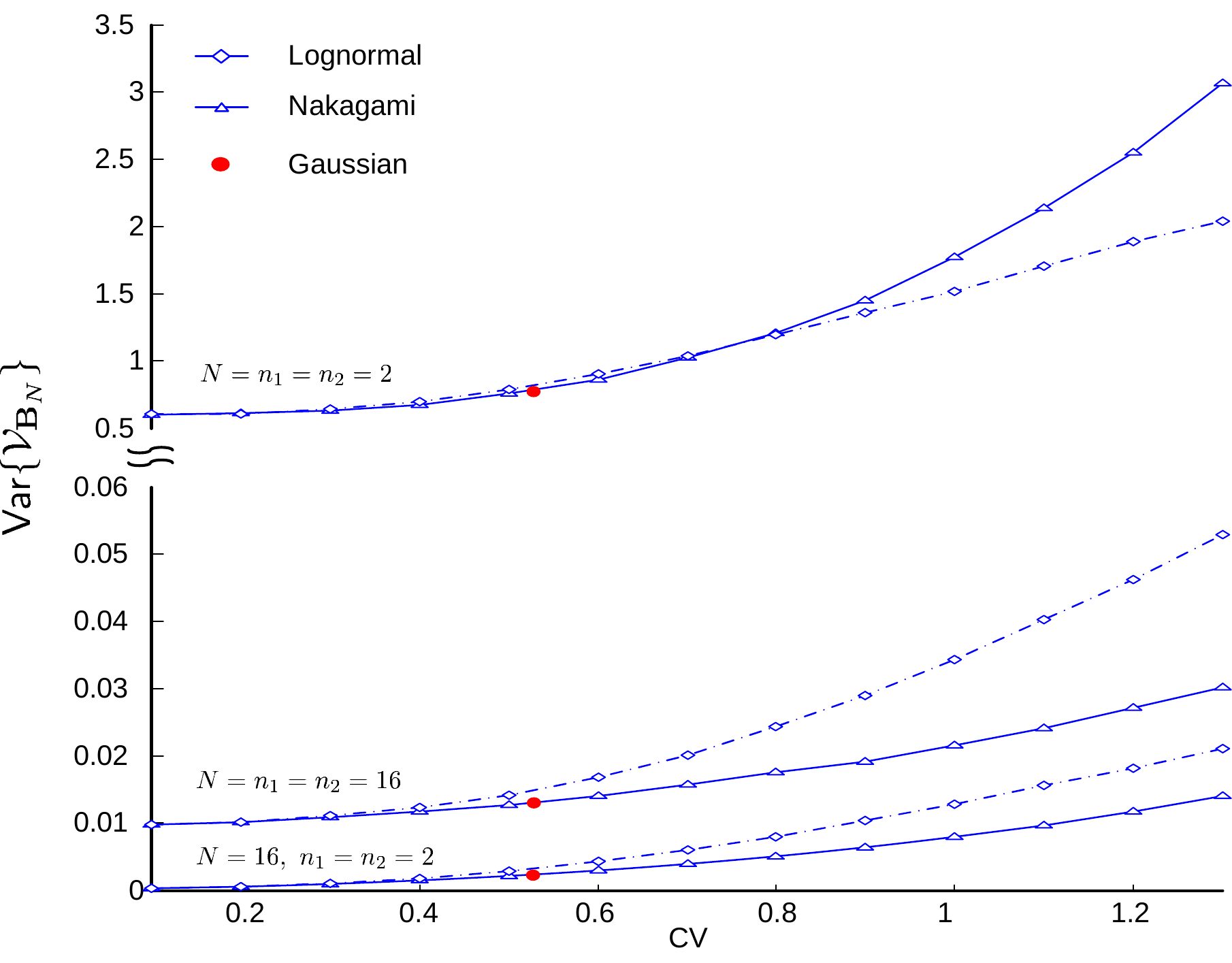} }%
\caption{ The empirical variance versus CV when $\bar\qH = \qzero$ and $1/\sigma^2= 30$ (dB). The above plot corresponds to a small system while the below plots correspond to the large system of interest.}\label{fig:var}
\end{center}
\end{figure}

Next, we evaluate the variance of $\calV_{\qB_N}(\sigma^2)$ by numerical simulations. Using the same parameters as those in Figure \ref{fig:simAnAnalK0}, Figure
\ref{fig:var} shows the empirical variances against CV when $1/\sigma^2= 30$ (dB). We observe that as the number of receive antennas grows large, the variance of
$\calV_{\qB_N}(\sigma^2)$ becomes small, or the mutual information approaches to a deterministic value in the large-system limit. The scenario with $N=16$ and
$n_1=n_2 =2$ particularly corresponds to typical SCNs, where the transmitter has a small number of antennas while the receiver is composed of large number of
antennas. This validates the practice of the deterministic approximation in the SCNs.

The CLT of $\calV_{\qB_N}(\sigma^2)$ has been recognized for different models by Moustakas {\em et al.}~\cite{Mou-03}, Taricco \cite{Tar-08}, and Hachem {\em et
al.}~\cite{Hac-08b,Hac-11}. Although the CLT is beyond the scope of this paper, we find it important to clarify some properties of the variance of
$\calV_{\qB_N}(\sigma^2)$. In the large system of interest (e.g., $N=n_1=n_2=16$ or $N=16,~n_1=n_2=2$), it is noted that the log-normal distribution undergoes the
highest variance. In addition, the curves of variance diverge as the CV value increases. Clearly, the CV does not provide a proper metrology neither for the mean
nor the variance of $\calV_{\qB_N}(\sigma^2)$ in the {\em large system limit}.\footnote{The insightful finding is due to A. Moustakas.} In this paper, we have
shown that $\calV_N(\sigma^2)$ depends on the second moment of the variables $X_{ij}^{(k)}$'s. As a consequence, the mean of the mutual information is invariant to
the type of fading distribution in the large system limit. Under a simpler model (where the correlation matrices are diagonal), it has been pointed out recently in
\cite{Hac-11} that the variance of mutual information depends not only on the second moment but also on the fourth moment of the variables $X_{ij}^{(k)}$'s. This
conjecture might be true in the SCNs of interest but at present, the required CLT to address the cases where the correlation matrices are generally nonnegative
definite and the channel entries are non-Gaussian is not at all understood.

Large-system results have been widely used to design the optimal input covariance \cite{Mou-03,Art-09,Dup-10,Cou-09}. With Gaussian channel entries, \cite{Dup-10}
showed that the input covariance design based on the large-system results can provide indistinguishable results to that achieved by stochastic programming (or the
Vu-Paulraj algorithm \cite{Vu-05}), even for the cases with a small number of antennas. It is important to know if such good characteristics still holds when the
channel entries are non-Gaussian. We clarify this property over the case with $\bar\qH = \qzero$. In this case, the ergodic mutual information is more sensitive to
the type of distribution and an iterative water-filling algorithm based on the large-system results can be used to obtain the asymptotic optimal input covariance
\cite[Table II]{Cou-09}.\footnote{If $\bar\qH \neq \qzero$, a similar iterative algorithm based on the large-system result was provided in \cite{Rie-08}.} For
stochastic optimization, the Vu-Paulraj algorithm based on the barrier method is used, in which the average mutual information and their first and second
derivatives are calculated by Monte-Carlo methods with $10^4$ trials. The number of iterations for the barrier method is set to 10. In Figure \ref{fig:optCov}, we
evaluate $\Ex\{\calV_{\qB_N}(\sigma^2)\}$ when the input covariance matrices are obtained using the large-system results and the stochastic optimization. As shown,
the asymptotic approach provides indistinguishable results to that achieved by stochastic programming in the case of non-Gaussian fading channels.

\begin{figure}
\begin{center}
\resizebox{4.5in}{!}{%
\includegraphics*{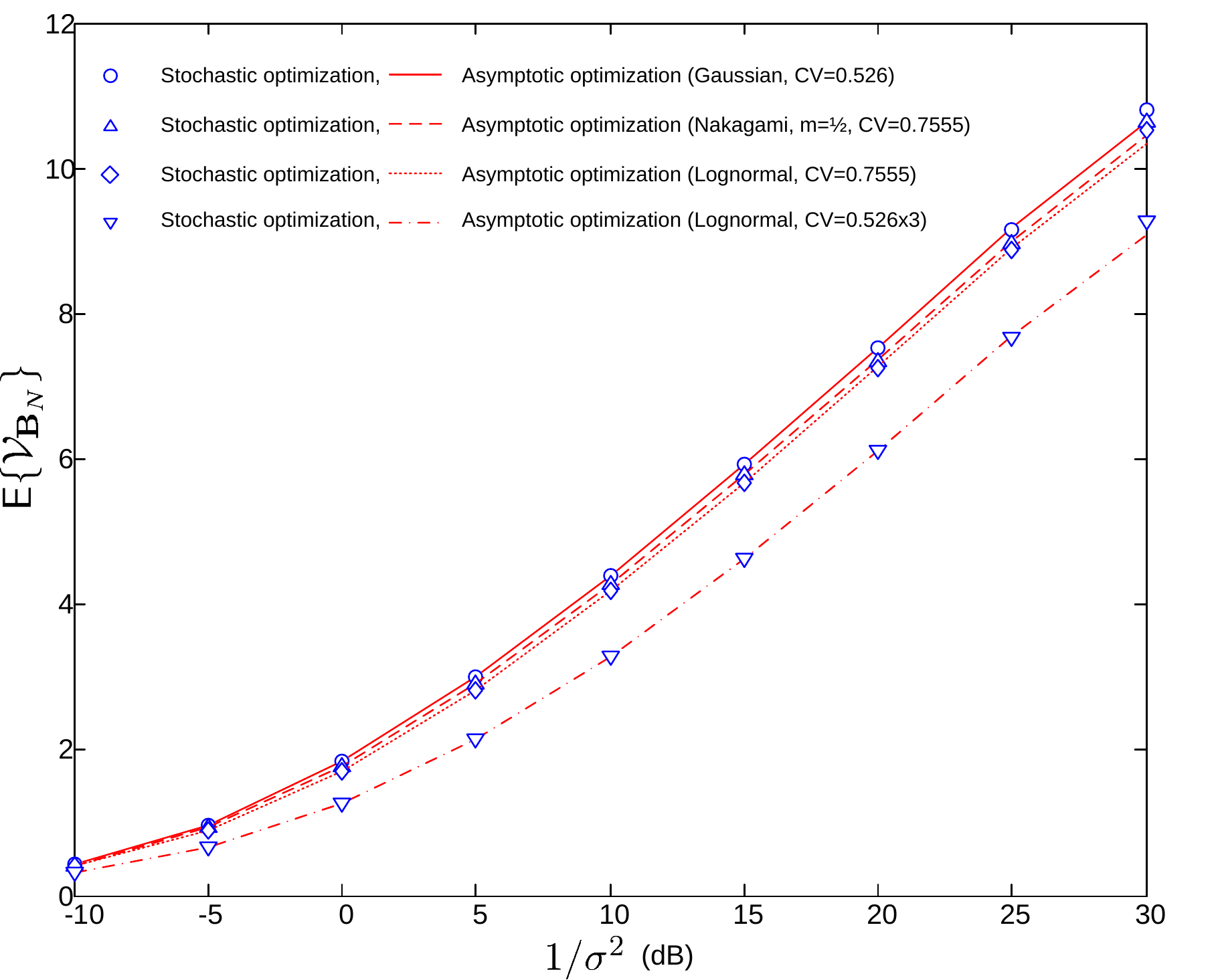} }%
\caption{Achievable rates versus SNR with $N=n_1=n_2=2$. The (red) lines plot the results based on the large-system results, and the marker points plot the results for stochastic optimization.}\label{fig:optCov}
\end{center}
\end{figure}

\section*{\sc V. Conclusion \label{Sec_Conclusion}}
This paper provided the deterministic equivalent of the LSD to deal with the channel matrices of SCNs where the entries of the MIMO channel matrix are \emph{no
longer} limited to be Gaussian distributed. Also, the correlation effects (caused by insufficient antenna spacing) and the LOS components (due to low antenna
heights) are included in the analysis. Using the deterministic equivalent of the LSD, we analyzed the Shannon transform of this class of large dimensional random
matrices and showed that the ergodic mutual information of the random matrices under investigation is invariant with respect to their distributions. As a
byproduct, we proved that the deterministic equivalents of the MIMO MAC in \cite{Cou-09} are true even if the entries of the channel matrix are non-Gaussian and
 $\qR_k$'s and $\qT_k$'s are nonnegative definite.

\section*{\sc Appendix A. Proof of Theorem \ref{mainTh_Stj}}
\subsection*{A.1 Truncation, Centralization, and Rescaling}
We begin the proof of Theorem \ref{mainTh_Stj} by replacing the entries of $\{\qX_k\}_{1\leq k \leq K}$ and that of the spectral decompositions of $\{ \qR_k,
\qT_k, \bar\qH_k\}_{1 \leq k \leq K}$ with truncated (and centralized) variables. It suffices to prove that the difference between the ESD of $\qB_N$ and the one
of truncated $\qB_N$ converges to zero with probability one because such convergence is equivalent to the convergence of their Stieltjes transforms.

We first follow a line similar to that in \cite[Section 4.3]{Bai-10} to truncate the spectral decompositions of $\{ \qR_k, \qT_k, \bar{\qH}_k\bar{\qH}_k^H\}_{1\leq
k\leq K}$. For any nonnegative definitive matrix $\qA \in \bbC^{m \times m}$, introduce its spectral decomposition and the corresponding truncation as follows:
\begin{equation}
\left\{\begin{aligned}
    \qA &= \qU_A \diag \left(\lambda_1(\qA),\dots,\lambda_{m}(\qA) \right) \qU_A^H, \\
    \qA^\alpha &= \qU_A \diag \left(\lambda_1(\qT_k)1_{(\lambda_1(\qA)\leq\alpha)},\dots,\lambda_{m}(\qA)1_{(\lambda_{m}(\qA)\leq\alpha)} \right) \qU_A^H,
\end{aligned}\right.
\end{equation}
where $\lambda_i(\qA)$ denotes the $i$th largest eigenvalues of $\qA$ and $\alpha>0$. Also, for the rectangular matrix $\bar\qH_k$, define its singular value
decomposition and the corresponding truncation version as follows:
\begin{equation}
\left\{\begin{aligned}
    \bar{\qH}_k &= \bar{\qU}_k \bar{\qSigma}_k \bar{\qV}_k^H, \\
    \bar{\qH}_k^\alpha &= \bar{\qU}_k \bar{\qSigma}_k^\alpha \bar{\qV}_k^H,
\end{aligned}\right.
\end{equation}
where $\bar{\qSigma}_k^\alpha$ is obtained from $\bar{\qSigma}_k$ with each singular value $\chi$ being replaced by $\chi 1_{(\chi \leq \alpha)}$. Let
$\chi_i(\bar{\qH}_k)$ denote the $i$th largest singular value of $\bar{\qH}_k$. By Lemma \ref{disIneq2} and iv) of Lemma \ref{rankIneq}, we have
\begin{align}
&\sup_{x} \left|F_{\qS + \sum_{k=1}^{K} \left(\bar\qH_k + \qR_k^\frac{1}{2} \qX_k \qT_k^\frac{1}{2}\right)\left(\bar\qH_k + \qR_k^\frac{1}{2} \qX_k \qT_k^\frac{1}{2}\right)^H}(x)
- F_{\qS + \sum_{k=1}^{K} \left(\bar\qH_k^\alpha + (\qR_k^\frac{1}{2})^{\alpha} \qX_k (\qT_k^\frac{1}{2})^{\alpha}\right)\left(\bar\qH_k^\alpha + (\qR_k^\frac{1}{2})^{\alpha} \qX_k (\qT_k^\frac{1}{2})^{\alpha}\right)^H}(x) \right| \notag \\
\leq& \frac{1}{N} \sum_{k=1}^{K} \left( \rank\left( \bar\qH_k-\bar\qH_k^\alpha \right)
+  \rank\left( \qR_k^{\frac{1}{2}}-(\qR_k^\frac{1}{2})^{\alpha} \right) + \rank\left( \qT_k^{\frac{1}{2}}-(\qT_k^\frac{1}{2})^{\alpha} \right) \right) \notag \\
=& \frac{1}{N} \sum_{k=1}^{K} \left( \sum_{i=1}^{N} 1_{\left(\chi_i(\bar{\qH}_k) > \alpha \right)} + \sum_{i=1}^{N} 1_{\left(\lambda_i(\qR_k^{1/2}) > \alpha \right)}
+ \sum_{i=1}^{n_k} 1_{\left(\lambda_i(\qT_k^{1/2}) > \alpha \right)} \right) \notag \\
=& \sum_{k=1}^{K} \left( F_{\bar\qH_k\bar\qH_k^H}((\alpha^2,\infty)) + F_{\qR_k}((\alpha^2,\infty)) + \frac{1}{\beta_k} F_{\qT_k}((\alpha^2,\infty)) \right). \label{eq:supF}
\end{align}
The right-hand side of the inequality above can be made arbitrary small if $\alpha$ is large enough by Assumption \ref{Ass3}. Therefore we can assume that the
eigenvalues of $\{ \qR_k, \qT_k, \bar{\qH}_k\bar{\qH}_k^H\}_{1\leq k\leq K}$ are bounded by a constant $\alpha$.

Next, we truncate and centralize the entries of $\qX_k$. As pointed out at Remark 1, the assumption $\Ex\{ X_{11}^{(k)}\}= 0$ can be removed from Theorem
\ref{mainTh_Stj} if $X_{ij}^{(k)}$'s have the same mean. For this reason, we do not make the zero mean assumption in the subsequent analysis. For each $k$, let
\begin{equation} \left\{\begin{aligned}
\hat{X}_{11}^{(k)}&=X_{11}^{(k)} 1_{(|X_{11}^{(k)}|\leq \varepsilon_{n}\sqrt{n_k})}\\
\tilde{X}_{11}^{(k)}&=\hat{X}_{11}^{(k)}-\Ex\{\hat{X}_{11}^{(k)}\},
\end{aligned}\right.
\end{equation}
where
\begin{equation}
    \varepsilon_n\downarrow 0, ~~~\mbox{and}~~~ \varepsilon_n^{-2} \, \Ex\left\{ |X_{11}^{(k)}|^2 \, 1_{(|X_{11}^{(k)}|>\varepsilon_{n}\sqrt{n_k})} \right\} \longrightarrow 0.
\end{equation}
Also, define $\hat{\qX}_k=[\frac{1}{\sqrt{n_k}}\hat{X}_{ij}^{(k)}] \in \bbC^{N \times n_k}$ and $\tilde{\qX}_k=[\frac{1}{\sqrt{n_k}} \tilde{X}_{ij}^{(k)}] \in
\bbC^{N \times n_k}$, and $\hat{\qB}_N$ and $\tilde{\qB}_N$ (obtained from $\qB_N$ with $X_{ij}^{(k)}$ replaced by $\hat{X}_{ij}^{(k)}$ and $\tilde{X}_{ij}^{(k)}$,
respectively). By Lemma \ref{disIneq2} and iv) of Lemma \ref{rankIneq}, we obtain
\begin{equation}\label{a1}
\sup_{x} \left| F_{\qB_N}(x)-F_{\hat{\qB}_N}(x) \right|\leq \frac{1}{N} \sum_{k=1}^K  \rank(\hat{\qX}_k-\qX_k)\leq \frac{1}{N} \sum_{k=1}^K\sum_{ij} 1_{(|X_{ij}^{(k)}|> \varepsilon_{n}\sqrt{n_k})} \xrightarrow{a.s.} 0,
\end{equation}
where the last step can be obtained in the same way as in \cite[Section 4.3.2]{Bai-10}. Repeating the first inequality in (\ref{a1}) with $F_{\qB_N}(x)$ replaced
by $F_{\tilde{\qB}_N}(x)$ yields\footnote{Note that $\rank( \hat\qX_k - \tilde \qX_k ) = \rank(\Ex\{\hat{\qX}_k\}) \leq 1$.}
\begin{equation}
\sup_{x} \left|F_{\tilde{\qB}_n}(x)-F_{\hat{\qB}_n}(x)\right|\stackrel{a.s.}\longrightarrow 0.
\end{equation}
Similarly, we may show that re-normalization of $\tilde{X}_{ij}^{(k)}$ does not affect the LSD of $F^{\tilde{\qB}_N}(x)$ as in \cite[Section 3.2]{Bai-10}.

Therefore, henceforth, we consider that Assumption \ref{Ass4} holds. For ease of reading, we recall this assumption here: For each of $N$, $X_{ij}^{(k)}$ are
i.i.d., and
\begin{subequations}
\begin{align}
 & \Ex \left\{X_{11}^{(k)}\right\} =0, ~ \Ex\left\{|X_{11}^{(k)}|^2\right\}=1,~|X_{11}^{(k)}| \leq \varepsilon_{n}\sqrt{n_k}, \label{eq:XassReCall}\\
 & \max_{k=1,\dots,K}\,\max\{\|\qR_k\|,\|\qT_k\|,\|\bar{\qH}_k\bar{\qH}_k^H\|\} \leq \alpha.
\end{align}
\end{subequations}
For convenience, we still use $\qX_k$, $\qT_k$, $\qR_k$, and $\bar{\qH}_k$ to denote those truncated and centralized matrices.

\subsection*{A.2 Proof of Step 1}
The aim in this subsection is to prove that
\begin{equation}\label{eq:a3}
    \Ex \left\{ \left| m_{\qB_N}(z) - \Ex \{m_{\qB_N}(z)\} \right|^{2p} \right\} = O\left(\frac{1}{N^p}\right) ~~~ \mbox{for any}~~~ p \geq 2,
\end{equation}
which, together with Borel-Cantelli's lemma, ensures Step 1. For ease of explanation, we prove the case with $K =1$ only but the similar procedure can be easily
extended to the case with $K \geq 1$. For this reason, we omit the index $k$ in the following procedure.

Let $\qx_j$ denote the $j$th column of $\qX$, $\qe_j$ be the column vector with the $j$th element being 1 and otherwise $0$, and set
\begin{equation}\label{a2}
\qX_{(j)}= \qX - \qx_j \qe_j^T.
\end{equation}
Furthermore, we find it useful to define
\begin{align}
m_{\qB_N}(z)&= \frac{1}{N} \tr\left( \qB_N - z \qI_{N} \right)^{-1},\\
m_{\qB_{(j)}}(z)&= \frac{1}{N} \tr\left( \qB_{(j)} - z \qI_{N} \right)^{-1},
\end{align}
where $\qB_{(j)} = (\qR^\frac{1}{2} \qX_{(j)} \qT^\frac{1}{2} + \bar{\qH})(\qR^\frac{1}{2} \qX_{(j)} \qT^\frac{1}{2} + \bar{\qH})^H+\qS$. Also, we use $\Ex_j$ to
denote conditional expectation given $\qx_{j+1},\dots,\qx_n$, so that $\Ex _{0}\{m_{\qB_N}(z)\} = m_{\qB_N}(z)$ and $\Ex_{n}\{m_{\qB_N}(z)\} = \Ex
\{m_{\qB_N}(z)\}$. Therefore, we have
\begin{align}
m_{\qB_N}(z) - \Ex \{m_{\qB_N}(z)\}
&= \sum_{j=1}^{n} [\Ex _{j-1}\{m_{\qB_N}(z)\} - \Ex _{j}\{m_{\qB_N}(z)\}] \notag\\
&= \sum_{j=1}^{n} [\Ex _{j-1} - \Ex _{j}]\{\tr\left( (\qB_N - z \qI_{N})^{-1} \right) - \tr\left( (\qB_{(j)} - z \qI_{N})^{-1} \right) \} \notag\\
&= \frac{1}{N} \sum_{j=1}^{n} [\Ex _{j} - \Ex _{j-1}]\{\gamma_{j1}+\gamma_{j2}+\gamma_{j3}+\gamma_{j4}+\gamma_{j5}\},\label{eq:difmm}
\end{align}
where
\begin{subequations}
\begin{align}
\gamma_{j1} &= - \qe_j^T \qT \qX_{(j)}^H \qR^{\frac{1}{2}} (\qB_{(j)}-z\qI)^{-1}(\qB_N-z\qI)^{-1} \qR^{\frac{1}{2}} \qx_j, \\
\gamma_{j2} &= - \qx_j^H \qR^{\frac{1}{2}} (\qB_N-z\qI)^{-1}(\qB_{(j)}-z\qI)^{-1} \qR^{\frac{1}{2}} \qX_{(j)} \qT \qe_j, \\
\gamma_{j3} &= - \qe_j^T \qT \qe_j \qx_j^H \qR^{\frac{1}{2}} (\qB_{(j)}-z\qI)^{-1}(\qB_N-z\qI)^{-1} \qR^{\frac{1}{2}} \qx_j, \\
\gamma_{j4} &= - \qe_j^T \qT^{\frac{1}{2}} \bar{\qH}^H (\qB_{(j)}-z\qI)^{-1}(\qB_N-z\qI)^{-1} \qR^{\frac{1}{2}} \qx_j, \\
\gamma_{j5} &= - \qx_j^H \qR^{\frac{1}{2}} (\qB_N-z\qI)^{-1}(\qB_{(j)}-z\qI)^{-1} \bar{\qH} \qT^{\frac{1}{2}} \qe_j.
\end{align}
\end{subequations}
In (\ref{eq:difmm}), we have used the resolvent identity (see Lemma \ref{lemma_Resolvent}), (\ref{a2}) and
\begin{multline}
\qB_N - \qB_{(j)}= \qR^{\frac{1}{2}} \qx_j \qe_j^T \qT \qX_{(j)}^H \qR^{\frac{1}{2}}+ \qR^{\frac{1}{2}} \qX_{(j)} \qT \qe_j \qx_j^H \qR^{\frac{1}{2}}\\
+\qR^{\frac{1}{2}} \qx_j \qe_j^T \qT \qe_j \qx_j^H \qR^{\frac{1}{2}}+\qR^{\frac{1}{2}} \qx_j \qe_j^T \qT^{\frac{1}{2}} \bar{\qH}^H + \bar{\qH} \qT^{\frac{1}{2}}
\qe_j \qx_j^H \qR^{\frac{1}{2}}.
\end{multline}

Since the mathematical treatments for $\gamma_{j4}$ and $\gamma_{j5}$ are similar, we here consider $\gamma_{j5}$ only. Starting from the Cauchy-Schwartz
inequality and then applying 3) of Lemma \ref{LemmaNorm}, we get
\begin{align}
|\gamma_{j5}| &\leq \| \qR^{\frac{1}{2}} (\qB_N-z\qI)^{-1}(\qB_{(j)}-z\qI)^{-1} \bar{\qH} \qT^{\frac{1}{2}} \qe_j \| \| \qx_j \|\notag\\
&\leq \| \qR^{\frac{1}{2}} \| \| (\qB_N-z\qI)^{-1} \| \| (\qB_{(j)}-z\qI)^{-1} \| \|\bar{\qH}\|  \|\qT^{\frac{1}{2}} \| \| \qx_j \|.
\end{align}
Lemma \ref{PansIneq} gives
\begin{equation} \label{eq:x_ord}
\Ex\left\{\|\qx_j\|^{2p}\right\}= O(1)~~\mbox{for any }p \geq 2.
\end{equation}
Note that $\|\qR\|$, $\|\qT\|$, $\|\bar{\qH}\|$, $\|(\qB_N-z\qI)^{-1}\|$ and $\|(\qB_{(j)}-z\qI)^{-1}\|$ are all bounded. Hence, we have
\begin{equation} \label{eq:gamma5leq1}
\Ex\left\{|\gamma_{j5}|^{2p}\right\}=O(1)~~\mbox{for any }p \geq 2.
\end{equation}
Then, by Lemma \ref{Burk2}, we can show that for any $p \geq 2$,
\begin{equation} \label{eq:gamma5leq3}
\Ex \left\{\left| \frac{1}{N} \sum_{j=1}^{n} [\Ex _{j} - \Ex _{j-1}]\{\gamma_{j5}\} \right|^{2p}\right\}\leq \frac{C_p}{N^{2p}} \Ex\left\{\left( \sum_{j=1}^{n}  |\gamma_{j5}|^{2} \right)^{p}\right\}\leq \frac{C C_p}{N^{p+1}} \sum_{j=1}^{n} \Ex\left\{ |\gamma_{j5}|^{2p}\right\}= O\left( \frac{1}{N^{p}} \right),
\end{equation}
where the second inequality follows from Lemma \ref{sumIneq} and the last equality is due to (\ref{eq:gamma5leq1}).

Next, we consider $\gamma_{j3}$. Let $\dist$ stand for the Euclidean distance between $z$ and $\bbR^+$. Since $\|(\qB_N-z\qI)^{-1}\|$ and
$\|(\qB_{(j)}-z\qI)^{-1}\|$ are both bounded by $\frac{1}{\dist}$, using Lemma \ref{LemmaNorm} gives
\begin{equation} \label{eq:gamma3leq1}
|\gamma_{j3}| \leq |\qe_j^T \qT \qe_j|  ~\|\qR^\frac{1}{2} (\qB_{(j)}-z\qI)^{-1}(\qB_N-z\qI)^{-1} \qR^\frac{1}{2}\| \|\qx_j\|^2\leq \frac{\|\qT\|\,\|\qR\|}{\dist^2} \|\qx_j\|^2.
\end{equation}
In addition, a simple application of Lemma \ref{PansIneq} gives $ \Ex\{|\gamma_{j3}|^p\}= O(1)$ for any $p \geq 2$. Then, applying the same arguments as in
(\ref{eq:gamma5leq3}), we have that for any $p \geq 2$,
\begin{equation} \label{eq:gamma3leq3}
\Ex\left\{ \left| \frac{1}{N} \sum_{j=1}^{n} [\Ex _{j} - \Ex _{j-1}]\{\gamma_{j3}\} \right|^{p}\right\}\leq \frac{C C_p}{N^{\frac{p}{2}+1}} \sum_{j=1}^{n} \Ex\left\{ |\gamma_{j3}|^{p}\right\}= O\left( \frac{1}{N^{\frac{p}{2}}} \right).
\end{equation}

As the procedures for $\gamma_{j1}$ and $\gamma_{j2}$ are similar, we take $\gamma_{j1}$ as an example. Using Lemma \ref{LemmaNorm}, we have
\begin{equation} \label{eq:gamma1leq1}
|\gamma_{j1}|  \leq  \|\qe_j^T \qT \qX_{(j)}^H\| ~\| \qR^\frac{1}{2} (\qB_{(j)}-z\qI)^{-1}(\qB -z\qI)^{-1} \qR^\frac{1}{2} ~\| \|\qx_j\|\leq \frac{\| \qR \|}{\dist^2} ~\|\qe_j^T \qT \qX_{(j)}^H\| ~\|\qx_j\|.
\end{equation}
Let $\acute{\qx}_{i}$ be the $i$th column vector of $\qX_{(j)}^H$. It is easily verified that
\begin{equation} \label{eq:xTeBound}
\Ex\left\{\acute{\qx}_{i}^H  \qT \qe_j \qe_j^T \qT \acute{\qx}_{i} \right\} \leq \frac{1}{n} \qe_j^T \qT^2 \qe_j \leq \frac{1}{n} \| \qT \|^2 \leq \frac{1}{n} \alpha^2.
\end{equation}
Then,
\begin{align}\label{eq:est_eTx}
\Ex \left\{ \|\qe_j^T \qT \qX_{(j)}^H\|^{2p}\right\} & \eqover{(a)} \Ex \left\{ \left| \sum_{i=1}^{n} \acute{\qx}_{i}^H  \qT \qe_j \qe_j^T \qT \acute{\qx}_{i}  \right|^p\right\}\notag\\
& = C_p \Ex\left\{  \left| \sum_{i=1}^{n}\Big(\acute{\qx}_{i}^H  \qT \qe_j \qe_j^T \qT \acute{\qx}_{i}
 - \Ex \left\{\acute{\qx}_{i}^H  \qT \qe_j \qe_j^T \qT \acute{\qx}_{i}\right\}\Big)\right|^p\right\}
 +C_p \left|\sum_{i=1}^{n} \Ex \left\{\acute{\qx}_{i}^H  \qT \qe_j \qe_j^T \qT \acute{\qx}_{i}\right\}\right|^p \notag \\
 &\leqover{(b)} C_p  \sum_{i=1}^{n} \Ex \left\{\Big|\acute{\qx}_{i}^H  \qT \qe_j \qe_j^T \qT \acute{\qx}_{i}
  - \Ex \left\{\acute{\qx}_{i}^H  \qT \qe_j \qe_j^T \qT \acute{\qx}_{i}\right\}\Big|^p\right\} \notag \\
  & ~~~ + C_p \left(\sum_{i=1}^{n} \Ex\left\{ \Big|\acute{\qx}_{i}^H  \qT \qe_j \qe_j^T \qT \acute{\qx}_{i}
 -   \Ex \left\{\acute{\qx}_{i}^H  \qT \qe_j \qe_j^T \qT \acute{\qx}_{i}\right\}\Big|^2\right\} \right)^\frac{p}{2}
 +C_p \alpha^{2p},
\end{align}
where (a) is due to the fact that $\qX_{(j)}^H \qX_{(j)} = \sum_{i=1}^{n} \acute{\qx}_{i} \acute{\qx}_{i}^H$, and (b) follows from Lemma \ref{lemma_eleIneq} and
(\ref{eq:xTeBound}). From Lemma \ref{sumIneq} and (\ref{eq:xTeBound}), we get
\begin{equation} \label{eq:xTeIneq}
\Ex \left\{\Big|\acute{\qx}_{i}^H  \qT \qe_j \qe_j^T \qT \acute{\qx}_{i}- \Ex \left\{\acute{\qx}_{i}^H  \qT \qe_j \qe_j^T \qT \acute{\qx}_{i}\right\}\Big|^p\right\}
\leq 2^{p-1} \left( \Ex \left\{\Big|\acute{\qx}_{i}^H  \qT \qe_j \qe_j^T \qT \acute{\qx}_{i}\Big|^p\right\} + \frac{1}{n^p} \alpha^{2p}\right).
\end{equation}
Substituting this into (\ref{eq:est_eTx}), we obtain
\begin{align}\label{eq:est_eTx2}
\Ex \left\{ \|\qe_j^T \qT \qX_{(j)}^H\|^{2p}\right\}
 &\leq C_p C' \sum_{i=1}^{n} \Ex\left\{\Big|\acute{\qx}_{i}^H  \qT \qe_j \qe_j^T \qT \acute{\qx}_{i}\Big|^p\right\}
 + C_p \left(\sum_{i=1}^{n} \Ex\left\{\Big|\acute{\qx}_{i}^H  \qT \qe_j \qe_j^T \qT \acute{\qx}_{i}\Big|^2\right\}\right)^{\frac{p}{2}} + C .
\end{align}
Moreover, since we know that
\begin{equation}
\| \qT \qe_j \|^2 = \qe_j^H \qT^2 \qe_j \leq \alpha^2,
\end{equation}
Lemma \ref{PansIneq} gives
\begin{equation}
\Ex \left\{ |\acute{\qx}_{k}^H  \qT \qe_j |^{2p}\right\}= O\left( \frac{1}{N} \right) ~~\mbox{for any }p \geq 1.
\end{equation}
From this and (\ref{eq:est_eTx2}), it follows that
\begin{equation} \label{eq:gamma1leq2}
\Ex\left\{  \|\qe_j^T \qT \qX_{(j)}^H\|^{2p}\right\}=O\left(1 \right).
\end{equation}
By applying the independence between $\qX_{(j)}$ and $\qx_j$ and using (\ref{eq:x_ord}), (\ref{eq:gamma1leq1}), and (\ref{eq:gamma1leq2}), we have
\begin{equation}
\Ex\left\{ |\gamma_{j1}|^{2p}\right\} \leq \frac{\| \qR \|^{2p}}{\dist^{4p}} ~\Ex\left\{ \|\qe_j^T \qT \qX_{(j)}^H\|^{2p}\right\} ~\Ex\left\{ \|\qx_j\|^{2p}\right\}= O\left(1\right).
\end{equation}
Therefore, using Lemma \ref{Burk2} with the above, we have, for any $p \geq 1$,
\begin{equation} \label{eq:gamma1leq3}
\Ex  \left\{\left| \frac{1}{N} \sum_{j=1}^{n} [\Ex _{j} - \Ex _{j-1}]\{\gamma_{j1}\} \right|^{2p}\right\}\leq \frac{C_p}{N^{2p}} \Ex\left\{\left( \sum_{j=1}^{n}  |\gamma_{j1}|^{2} \right)^{p}\right\}= O\left( \frac{1}{N^{p}} \right).
\end{equation}
(\ref{eq:a3}) then follows from (\ref{eq:gamma5leq3}), (\ref{eq:gamma3leq3}), and (\ref{eq:gamma1leq3}). The proof is complete.

\subsection*{A.3 Proof of Step 2}

To begin with, recall the definition:
\begin{align}
\qB_N&= \qS + \left(\qR^{\frac{1}{2}} \qX \qT^{\frac{1}{2}} + \bar{\qH}\right)\left(\qR^{\frac{1}{2}} \qX \qT^{\frac{1}{2}} + \bar{\qH}\right)^H,\\
\calqB_N&= \qS + \left(\qR^{\frac{1}{2}} \calqX \qT^{\frac{1}{2}} + \bar{\qH}\right)\left(\qR^{\frac{1}{2}} \calqX \qT^{\frac{1}{2}} + \bar{\qH}\right)^H,
\end{align}
where $\qX$ and $\calqX$ are matrices with entries satisfying (\ref{eq:XassReCall}) but $\calqX$ is Gaussian. The aim here is to prove
\begin{equation}
\left|  \Ex \{ m_{\qB_N}(z) \} - \Ex \{ m_{\calqB_N}(z) \} \right| = O\left(\varepsilon_n \right).
\end{equation}
As before, we will prove the case $K =1$ only and drop the unnecessary index $k$ in the sequel.

The strategy is to use Lemma \ref{Lindeberg}, the Lindeberg principle \cite[Theorem 2]{Kor-11}. As pointed out at the end of the first paragraph of Appendix A,
$\Ex \{X_{ij}\} = \Ex \{\calX_{ij}\}=0$. Also we have $\Ex \{|X_{ij}|^2\} =\Ex \{|\calX_{ij}|^2\} = 1$. Therefore $a_i = b_i = 0~\forall i$ in (\ref{eq:LindeAB}).
We next evaluate the second and third lines of (\ref{eq:LindeCheck}). To achieve this, we need to take the derivatives with respect to the real and imaginary parts
of the $(i,j)$th entries of $\qX$, respectively. Because the real and imaginary parts of $X_{ij}$  are independent, all the results established in the real case
can be directly applied for the complex case. Thus, without loss of generality, we deal with $\qX$ and $\calqX$ with {\em real} entries only in order to present
the formulas in a compact and succinct way.

For ease of exposition, we define
\begin{equation}
f(\qA) \triangleq \frac{1}{N} \tr \left( \qS + \left(\qR^{\frac{1}{2}} \qA \qT^{\frac{1}{2}} + \bar{\qH}\right)\left(\qR^{\frac{1}{2}} \qA \qT^{\frac{1}{2}} + \bar{\qH}\right)^H -z\qI \right)^{-1},
\end{equation}
where $\qA$ is any matrix such that the product $\qR^{\frac{1}{2}} \qA \qT^{\frac{1}{2}}$ exists. As such, we have $m_{\qB_N}(z) = f(\qX)$ and $m_{\calqB_N}(z)
=f(\calqX)$. Moreover, to apply (\ref{eq:LindeCheck}), $\qA$ will take the form $\qA(r,c,s) = [A_{ij}(r,c,s)] \in \bbC^{N \times n}$ with
\begin{equation} \label{eq:defArcs}
 A_{ij}(r,c,s) = \left\{\begin{array}{cl}
\frac{X_{ij}}{\sqrt{n}},  & \mbox{if}~ i<r, ~\mbox{or}~ i=r ~\mbox{and}~ j < c, \\
s,                        & \mbox{if}~ i=r, ~\mbox{and}~ j=c, \\
\frac{\calX_{ij}}{\sqrt{n}},  & \mbox{otherwise}.
\end{array}\right.
\end{equation}
Further, let $\qG \triangleq \qS + \left(\qR^{\frac{1}{2}} \qA \qT^{\frac{1}{2}} + \bar{\qH}\right)\left(\qR^{\frac{1}{2}} \qA \qT^{\frac{1}{2}} +
\bar{\qH}\right)^H$, denote the partial derivative with respect to $A_{ij}$ by $\partial_{ij}$, and let $\qE_{ij}$ be the matrix with a 1 in the $(i,j)$th position
and 0's elsewhere. To get the third-fold derivative of $f(\qA)$, we rely on the following differentiation formulas:
\begin{subequations}\label{eq:difFormula}
\begin{align}
\partial_{ij}(\qG-z\qI)^{-1}&= -(\qG-z\qI)^{-1} (\partial_{ij}\qG) (\qG-z\qI)^{-1}, \\
\partial_{ij}\qG &=\left(\qR^{1/2} \qE_{ij} \qT^{\frac{1}{2}} \right) \left(\qR^{\frac{1}{2}} \qA \qT^{\frac{1}{2}} + \bar{\qH}\right)^H+ \left(\qR^{\frac{1}{2}} \qA \qT^{\frac{1}{2}} + \bar{\qH}\right) \qT^{\frac{1}{2}} \qE_{ji} \qR^\frac{1}{2},\\
\partial_{ij}^{2}\qG &= 2 T_{jj} \qR^\frac{1}{2} \qE_{ii} \qR^\frac{1}{2},\\
\partial_{ij}^{3}\qG &= 0.
\end{align}
\end{subequations}
By (\ref{eq:difFormula}), one can easily show that
\begin{subequations}
\begin{align}
\partial_{ij}f(\qA) =& - \frac{1}{N} \tr\left( (\partial_{ij}\qG) (\qB_N-z\qI)^{-2} \right), \\
\partial_{ij}^{2}f(\qA) =& \frac{2}{N} \tr\left( (\partial_{ij}\qG) (\qB_N-z\qI)^{-1} (\partial_{ij}\qG) (\qB_N-z\qI)^{-2} \right)- \frac{1}{N} \tr\left( (\partial_{ij}^2\qG) (\qB_N-z\qI)^{-2} \right),\\
\partial_{ij}^{3}f(\qA) =& - \frac{6}{N} \tr\left( (\partial_{ij}\qG) (\qB_N-z\qI)^{-1} (\partial_{ij}\qG) (\qB_N-z\qI)^{-1} (\partial_{ij}\qG) (\qB_N-z\qI)^{-2} \right)\notag\\
                          &+ \frac{3}{N} \tr\left( (\partial_{ij}^2\qG) (\qB_N-z\qI)^{-1} (\partial_{ij}\qG) (\qB_N-z\qI)^{-2} \right)\notag\\
                          &+ \frac{3}{N} \tr\left( (\partial_{ij}\qG) (\qB_N-z\qI)^{-1} (\partial_{ij}^2\qG) (\qB_N-z\qI)^{-2} \right).\label{eqn:72c}
\end{align}
\end{subequations}

Now, we provide a bound for each of the three terms of $\partial_{ij}^{3}f(\qA)$ (see (\ref{eqn:72c}) above). The first term of $\partial_{ij}^{3}f(\qA)$ can be
bounded by
\begin{align}
&\left| \tr\left( (\partial_{ij}\qG) (\qB_N-z\qI)^{-1} (\partial_{ij}\qG) (\qB_N-z\qI)^{-1} (\partial_{ij}\qG) (\qB_N-z\qI)^{-2} \right)\right|\notag\\
\leqover{(a)} & \| ((\partial_{ij}\qG) (\qB_N-z\qI)^{-1})^2 \|_{\sfF} ~\| (\partial_{ij}\qG) (\qB_N-z\qI)^{-2} \|_{\sfF} \\
\leqover{(b)} & \frac{1}{\dist} \| (\partial_{ij}\qG) (\qB_N-z\qI)^{-1} \|_{\sfF}^3  \\
\leqover{(c)} & \frac{1}{\dist^4} \| (\partial_{ij}\qG) \|_{\sfF}^3,
\end{align}
where (a) follows from 1)-i) of Lemma \ref{LemmaNorm} and the remaining two inequalities, (b) and (c), follow from 1)-ii) and 1)-iii) of Lemma \ref{LemmaNorm}. By
1)-i) and 1)-ii) of Lemma \ref{LemmaNorm}, the second and third terms of $\partial_{ij}^{3}f(\qA)$ can be bounded by
\begin{align}
& \left|\tr\left( (\partial_{ij}^2\qG) (\qB_N-z\qI)^{-1} (\partial_{ij}\qG) (\qB_N-z\qI)^{-2} \right)\right| \notag\\
\leq & \frac{1}{\dist^3} \| (\partial_{ij}^2\qG) \|_{\sfF} ~ \| (\partial_{ij}\qG) \|_{\sfF}=\frac{2}{\dist^3} \| T_{jj} \qR^\frac{1}{2} \qE_{ii} \qR^\frac{1}{2} \|_{\sfF} ~\| (\partial_{ij}\qG)
\|_{\sfF}.
\end{align}
Therefore, to estimate $|\partial_{ij}^{3}f(\qA)|$, we note that $\|T_{jj}\qR^\frac{1}{2}\qE_{ii}\qR^\frac{1}{2}\|_{\sfF}=T_{jj}R_{ii}\leq\|\qT\|\|\qR\|\leq
\alpha^2$ and
\begin{align} \label{eq:RETHR_Ineq}
\| \qR^\frac{1}{2} \qE_{ij} \qT^\frac{1}{2} \bar\qH^H\|_{\sfF}^2
&= R_{ii}\,\tr\left( \qe_j^T \qT^\frac{1}{2} \bar\qH^H \bar\qH \qT^\frac{1}{2} \qe_j \right) \notag\\
&\leqover{(a)} \|\qR\| \| \qT^\frac{1}{2} \bar\qH^H \bar\qH \qT^\frac{1}{2} \| \tr\left(\qe_j \qe_j^T\right) \notag \\
&\leqover{(b)} \| \qT \| \| \qR\| \| \bar\qH^H \bar\qH \| \leq \alpha^3,
\end{align}
where (a) follows from 1)-iv) and 2) of Lemma \ref{LemmaNorm}, and (b) follows from 3) of Lemma \ref{LemmaNorm}. As a result,
\begin{align} \label{eq:parGijBound}
\| (\partial_{ij}\qG) \|_{\sfF}^2 ~~&\leqover{(a)} 4 \left( \| \qR^\frac{1}{2} \qE_{ij} \qT \qA^T\qR^\frac{1}{2} \|_{\sfF}^2+\| \qR^\frac{1}{2} \qE_{ij} \qT^\frac{1}{2} \bar\qH^H \|_{\sfF}^2 \right) \notag\\
&\leqover{(b)} 4 \left( \tr \left( \qR^\frac{1}{2} \qE_{ij} \qT \qA^T \qR \qA \qT \qE_{ji} \qR^\frac{1}{2} \right)+ \alpha^3 \right) \notag \\
&\leqover{(c)} 4 \left( \| \qR \|^{2} \tr \left( \qe_j^T  \qT \qA^T \qA \qT \qe_j \right) + \alpha^3 \right) \notag \\
&=  4 \left( \| \qR \|^2 \sum_{i=1}^{N} |\acute{\qa}_i \qT \qe_j|^2 + \alpha^3 \right),
\end{align}
where (a) follows from the triangle inequality of the Frobenius norm and Lemma \ref{sumIneq}, (b) follows from (\ref{eq:RETHR_Ineq}), (c) follows from 1)-iv) and
2) of Lemma \ref{LemmaNorm}, and $\acute{\qa}_i$ represents the $i$th row vector of $\qA$.

Recalling the definition of $\qA(r,c,s)$ in (\ref{eq:defArcs}), when $i\neq r$, a direct application of Lemma \ref{lemma_eleIneq} yields
\begin{align}
\Ex\left\{\left| \acute{\qa}_i \qT \qe_j \right|^{2p}\right\}&=\Ex\left\{\left| \sum_{k=1}^{n} A_{ik}T_{kj}  \right|^{2p}\right\}\notag\\
&\leq C_p  \left( \sum_{k=1}^{n} \Ex\left\{ |A_{ik}|^{2p} |T_{kj}|^{2p}\right\}+\left( \sum_{k=1}^{n} \Ex\left\{|A_{ik}|^2 |T_{kj}|^2\right\}\right)^{p} \right).\label{eq:bnd_ea}
\end{align}
When $i=r$, similarly, we have
\begin{equation}\label{eq:bnd_ea1}
\Ex\left\{\left| \acute{\qa}_i \qT \qe_j \right|^{2p}\right\} \leq C_p  \left( \sum_{k\neq c}^{n} \Ex\left\{|A_{ik}|^{2p} |T_{kj}|^{2p}\right\}+\left( \sum_{k\neq c}^{n} \Ex \left\{|A_{ik}|^2 |T_{kj}|^2\right\} \right)^{p} \right) + C_p |s|^{2p}.
\end{equation}
From the definition of $\qA(r,c,s)$ in (\ref{eq:defArcs}), we get
\begin{equation} \label{eq:AijIneq}
\Ex\{|A_{ij}|^{2p}\} \leq \left\{\begin{array}{cl}
\frac{\varepsilon_n^{2p-2}}{n},                           & \mbox{if}~ i<r, ~\mbox{or}~ i=r ~\mbox{and}~ j < c, \\
|s|^{2p},                                & \mbox{if}~ i=r, ~\mbox{and}~ j=c, \\
\frac{\prod_{l=1}^{p} (2l-1)}{n^p},  & \mbox{otherwise}.
\end{array}\right.
\end{equation}
Then, the above gives the simple bound $\Ex\{|A_{ij}|^{2p}\} \leq \frac{C}{n}$ for $i\neq r$ and $j\neq c$. Note that
\begin{equation} \label{eq:TklNorm}
    \sum_{k=1}^{n} |T_{kj}|^{2p} \leq   \left(\sum_{k=1}^{n} |T_{kj}|^{2}\right)^{p} = \left(\qe_j^T \qT^2 \qe_j \right)^{p}
    \leq \| \qT \|^{2p}
    \leq \alpha^{2p}.
\end{equation}
From (\ref{eq:AijIneq}) and (\ref{eq:TklNorm}), we get, when $i\neq r$,
\begin{equation}
\sum_{k=1}^{n} \Ex\left\{ |A_{ik}|^{2p} |T_{kj}|^{2p}\right\} \leq \frac{C}{n} \sum_{k=1}^{n} |T_{kj}|^{2p} \leq \frac{C \alpha^p}{n} = O \left( \frac{1}{n} \right),
\end{equation}
and similarly, when $i = r$,
\begin{equation}
\sum_{k\neq c}^{n} \Ex\left\{ |A_{ik}|^{2p} |T_{kj}|^{2p}\right\} = O \left( \frac{1}{n} \right).
\end{equation}
Therefore, using (\ref{eq:bnd_ea}) and (\ref{eq:bnd_ea1}) with the above bounds, we have
\begin{equation}
\Ex\left\{  |\acute{\qa}_i \qT \qe_j|^{4}\right\}=\left\{\begin{array}{cl}
O \left( \frac{1}{n} \right) + |s|^{4}, & \mbox{if }i = r,\\
O \left( \frac{1}{n} \right),  & \mbox{otherwise},
\end{array}\right.
\end{equation}
which, together with Lemma \ref{lemma_eleIneq}, ensures that
\begin{equation} \label{eq:aTeBound}
\Ex\left\{ \left|\sum_{i\neq r}^{N}|\acute{\qa}_i \qT \qe_j|^2-\Ex \left\{|\acute{\qa}_i \qT \qe_j|^2\right\} \right|^2\right\}\leq
C\sum_{i\neq r}^{N}\Ex\left\{ \left||\acute{\qa}_i \qT \qe_j|^2-\Ex\left\{ |\acute{\qa}_i \qT \qe_j|^2\right\} \right|^2\right\}= O(1).
\end{equation}
Combining everything together, we get
\begin{align}
\Ex\left\{ | \partial_{rc}^{3}f(\qA(r,c,s)) |\right\}
    &\leqover{(a)} \frac{C}{N} \Ex\left\{ \left( C' + \sum_{i=1}^{N} |\acute{\qa}_i \qT \qe_j|^2 \right)^{\frac{3}{2}}\right\} \notag \\
    &\leqover{(b)} \frac{C}{N} \left( C' + |s|^3 + \Ex\left\{ \left( \sum_{i\neq r}^{N} |\acute{\qa}_i \qT \qe_j|^2 \right)^{\frac{3}{2}}\right\} \right) \notag \\
    &\leq \frac{C}{N} \left( C' + |s|^3+ \Ex\left\{ \left( \sum_{i\neq r}^{N} |\acute{\qa}_i \qT \qe_j|^2-\Ex\left\{ |\acute{\qa}_i \qT \qe_j|^2\right\} \right)^{\frac{3}{2}}\right\} \right) +  \frac{C}{N}\left(\sum_{i\neq r}^{N} \Ex\left\{ |\acute{\qa}_i \qT \qe_j|^2\right\}\right)^{\frac{3}{2}} \notag \\
    &\leqover{(c)} \frac{C}{N} ~ \left( C' + |s|^3 + \left( \Ex\left\{ \left|\sum_{i\neq r}^{N}    |\acute{\qa}_i \qT \qe_j|^2-\Ex\left\{ |\acute{\qa}_i \qT \qe_j|^2\right\}\right|^2\right\}\right)^{\frac{3}{4}} \right)+\frac{C}{N}\left(\sum_{i\neq r}^{N} \Ex\left\{ |\acute{\qa}_i \qT \qe_j|^2\right\}\right)^{\frac{3}{2}} \notag \\
    &\leqover{(d)} \frac{C}{N} ~ \left( C' + |s|^3 \right), \label{eq:diffFinal}
\end{align}
where (a) follows from (\ref{eq:parGijBound}), (b) follows from Lemma \ref{sumIneq}, (c) is due to the fact that $(\Ex \{|\cdot|^p\})^\frac{1}{p}$ is a
nondecreasing function of $p$, and (d) follows from (\ref{eq:aTeBound}).

Finally, we can evaluate the second and third lines of (\ref{eq:LindeCheck}). Using (\ref{eq:LindeCheck}) and (\ref{eq:diffFinal}), we have\footnote{Note that
$\Re\{f(\qA)\}$ is a smooth function and $| \partial_{ij}^{p} \Re\{f(\qA)\} | \leq | \partial_{ij}^{p} f(\qA) |$ for each $p$.}
\begin{align}
    \left|\Ex \{ \Re\{m_{\qB_N}(z)\} \} - \Ex \{ \Re\{m_{\calqB_N}(z)\} \}\right|
    &\leq \frac{C}{N} \sum_{i=1}^{N} \sum_{j=1}^{n} \left(
    \Ex \left\{ \int_{0}^{|X_{ij}|/\sqrt{n}} (C'+|s|^3)\left( \frac{X_{ij}}{\sqrt{n}} - s \right)^2 ds \right\} \right. \notag\\
    & ~~~~~~~~~~~~~~~~~
    \left.
    + \Ex \left\{ \int_{0}^{|\calX_{ij}|/\sqrt{n}} (C'+|s|^3)\left( \frac{\calX_{ij}}{\sqrt{n}} - s \right)^2 ds \right\}
    \right) \notag\\
    &\leq \frac{C}{N} \sum_{i=1}^{N} \sum_{j=1}^{n} \left(
    \frac{C'}{3} \Ex \left\{\left(\frac{|X_{ij}|}{\sqrt{n}}\right)^3\right\} + \frac{1}{60} \Ex \left\{\left(\frac{|X_{ij}|}{\sqrt{n}}\right)^6\right\}  \right. \notag \\
    & ~~~~~~~~~~~~~~~~~ \left.
    + \frac{C'}{3}\Ex \left\{\left(\frac{|\calX_{ij}|}{\sqrt{n}}\right)^3\right\} + \frac{1}{60} \Ex \left\{ \left(\frac{|\calX_{ij}|}{\sqrt{n}}\right)^6\right\}
    \right)= O(\varepsilon_n).
\end{align}
The quantity $|\Ex \{ \Im\{m_{\qB_N}(z)\} \} - \Ex \{ \Im\{m_{\calqB_N}(z)\} \}|$ also admits the same upper bound. Therefore, we finish Step 2.

\section*{\sc Appendix B. Existence and Uniqueness}

In this appendix, we will consider existence and uniqueness of the solution to (\ref{eq:fixedPoint}).

\subsection*{Appendix B.1 Existence}
As pointed out in the paragraphs after Theorem \ref{Theorem_Hac07}, Formulas (\ref{eq:Psi})--(\ref{eq:tPhi}) can be obtained from those of \cite[Theorems 2.4 and
2.5]{Hac-07} by recovering the effect from the eigenvectors $\qU_k$'s. Therefore existence of $e_i(z)$ follows from that of the corresponding solution $\psi_i(z)$
of \cite[Theorems 2.4]{Hac-07}.

\subsection*{Appendix B.2 Uniqueness}

In fact, uniqueness of $e_i(z)$ also follows immediately from that of $\psi_i(z)$ in \cite[Theorems 2.4]{Hac-08b}. However, here we provide an alternative proof,
which is inspired by \cite{Dup-10,Cou-09}.

For the reader's convenience, we recall the notation introduced in Theorem \ref{mainTh_uniq}:
\begin{subequations} \label{eq:fixedPointRecall}
\begin{align}
e_i(z) &= \frac{1}{N} \tr\left( \qR_i \qPsi(z)\right),\\
\tilde{e}_i(z) &= \frac{1}{n_i} \tr\left( \qT_i \langle \tilde{\qPsi}(z) \rangle_{i}\right),
\end{align}
\end{subequations}
where
\begin{subequations} \label{eq:fixedPoint2Recall}
\begin{align}
    \qPsi(z) &= \left( \qPhi(z)^{-1} - z \bar{\qH} \tilde{\qPhi}(z) \bar{\qH}^H \right)^{-1},  \\
    \tilde{\qPsi}(z) &= \left( \tilde{\qPhi}(z)^{-1} - z \bar{\qH}^H \qPhi(z) \bar{\qH} \right)^{-1},  \\
    \qPhi(z) &= \frac{-1}{z}\left( -\frac{1}{z}\qS + \sum_{i=1}^{K} \tilde{e}_i(z)\qR_i + \qI_N \right)^{-1}, \\
    \tilde{\qPhi}(z)         &= \diag\left( \tilde{\qPhi}_1(z),\dots, \tilde{\qPhi}_K(z) \right),
\end{align}
\end{subequations}
with $\tilde{\qPhi}_i(z) \triangleq -\frac{1}{z} (\qI_{n_i} + \beta_i e_i(z)\qT_i)^{-1}$. Let $\qe_z = [z e_1(z) \cdots z e_K(z)]^T $, $\tilde\qe_z =
[z\tilde{e}_1(z) \cdots z \tilde{e}_K(z)]^T$, $\qPhi_z = z \qPhi(z)$, $\tilde{\qPhi}_z = z \tilde{\qPhi}(z)$, $\qPhi_{zi} = z \qPhi_i(z)$, and $\tilde{\qPhi}_{zi}
= z \tilde{\qPhi}_i(z)$. To facilitate our notations, we, henceforth, denote by $\qPsi = \qPsi(z)$, $\tilde{\qPsi} = \tilde{\qPsi}(z)$, $\qPhi = \qPhi(z)$,
$\tilde{\qPhi} = \tilde{\qPhi}(z)$, $\tilde{\qPhi}_i = \tilde{\qPhi}_i(z)$. Suppose that $\{ e_i^{\circ}(z),\tilde{e}_i^{\circ}(z)\}_{1 \leq i \leq K}$ are another
solutions satisfying (\ref{eq:fixedPointRecall}) and let $\qPsi^{\circ}$, $\tilde{\qPsi}^{\circ}$, $\qPhi^{\circ}$, $\tilde{\qPhi}^{\circ}$, $\qPhi_z^{\circ}$,
$\tilde{\qPhi}_z^{\circ}$, $\qe_z^{\circ}$, $\tilde\qe_z^{\circ}$ be the matrices/vectors by replacing the entries $e_i(z)$'s and $\tilde{e}_i(z)$'s in $\qPsi$,
$\tilde{\qPsi}$, $\qPhi$, $\tilde{\qPhi}$, $\qPhi_z$, $\tilde{\qPhi}_z$, $\qe_z$, $\tilde\qe_z$ with $e_i^{\circ}(z)$'s and $\tilde{e}_i^{\circ}(z)$'s
respectively. We prove the uniqueness of $\qe_z$ and $\tilde\qe_z$ by showing that $\qe_z - \qe_z^{\circ} = \qzero$ and $\tilde\qe_z - \tilde\qe_z^{\circ} =
\qzero$.

Denote $\underline{\qT}_i \triangleq \diag(\qzero_{n_1},\dots,\qzero_{n_{i-1}}, \qT_i, \qzero_{n_{i+1}},\dots, \qzero_{n_K})$ for $i=1,\dots,K$. To simplify notation, we let
\begin{equation} \label{eq:defU}
\begin{aligned}
    u_{1,ij} &= \frac{1}{N} \tr\left(\qR_i \qPsi \qR_j \qPsi^H \right), \\
    u_{2,ij} &= \frac{\beta_j}{N} \tr\left(\qR_i \qPsi \bar{\qH}_j \tilde{\qPhi}_{zj}\qT_j\tilde{\qPhi}_{zj}^H \bar{\qH}_j^H\qPsi^H \right), \\
    v_{1,ij} &= \frac{\beta_j}{n_j} \tr \left( \underline{\qT}_i \tilde{\qPsi} \underline{\qT}_j\tilde{\qPsi}^H \right), \\
    v_{2,ij} &= \frac{1}{n_j} \tr\left( \underline{\qT}_i \tilde{\qPsi} \bar{\qH}^H \qPhi_z \qR_j \qPhi_z^H \bar{\qH}\tilde{\qPsi}^H \right).
\end{aligned}
\end{equation}
Moreover, let
\begin{equation} \label{eq:defGamma}
    \qGamma = \left[
          \begin{array}{cccc}
          \qGamma_{11} & \qzero & \qzero & \qGamma_{12}\\
          \qzero & \qGamma_{11} & |z|^2\qGamma_{12} & \qzero\\
          \qzero & \qGamma_{21} & \qGamma_{22} & \qzero\\
          |z|^2\qGamma_{21} & \qzero & \qzero & \qGamma_{22}\\
          \end{array}
          \right]
\end{equation}
with $\qGamma_{11} = [\Gamma_{11,ij}] \in \bbC^{K \times K}, ~\qGamma_{12} = [\Gamma_{12,ij}] \in \bbC^{K \times K}, ~\qGamma_{21} = [\Gamma_{21,ij}] \in \bbC^{K
\times K}, ~\qGamma_{22} = [\Gamma_{22,ij}] \in \bbC^{K \times K}$, and
\begin{align*}
    \Gamma_{11,ij}&=
    \left\{
    \begin{aligned}
    &0, & & i=j, \\
    &\frac{ u_{2,ij} }{1- u_{2,ii}},
    & & i \neq j,
    \end{aligned}
    \right.
    & \Gamma_{12,ij} &=
    \frac{ u_{1,ij} }{1- u_{2,ii}}, \\
    \Gamma_{22,ij} &=
    \left\{
    \begin{aligned}
    &0, & & i=j, \\
    &\frac{ v_{2,ij} }{1- v_{2,ii}},
    & & i \neq j,
    \end{aligned}
    \right.
    & \Gamma_{21,ij} &=
    \frac{ v_{1,ij} }{1- v_{2,ii}}.
\end{align*}
Similarly, let $\qGamma^{\circ}$ as well as $\qGamma_{11}^{\circ}$, $\qGamma_{12}^{\circ}$, $\qGamma_{21}^{\circ}$, and $\qGamma_{22}^{\circ}$ be the matrices by
replacing $\qPsi$, $\tilde{\qPsi}$, $\qPhi$, and $\tilde{\qPhi}$ with $\qPsi^{\circ H}$, $\tilde{\qPsi}^{\circ H}$, $\qPhi^{\circ H}$, and $\tilde{\qPhi}^{\circ
H}$ respectively.

Now, write $e_i(z) = e_{i,1}(z) + {\sf j} e_{i,2}(z)$, $\tilde{e}_i(z) = \tilde{e}_{i,1}(z) + {\sf j} \tilde{e}_{i,2}(z)$ and $z \equiv z_1 + {\sf j} z_2$. A
direct calculation then yields
\begin{align}
    &\Im\{ze_i(z)\} = \Im\left\{ \frac{1}{N}\tr\left(\qR_i \qPsi \left(z\qPsi^{-H}\right) \qPsi^H \right)\right\}\notag \\
    \eqover{(a)}& \Im\left\{\frac{1}{N}\tr\left(\qR_i \qPsi \left( z\qS - \sum_{j=1}^{K} |z|^2\tilde{e}_j^*(z)\qR_j - \sum_{j=1}^{K} z \bar{\qH}_j \tilde{\qPhi}_{zj}^H \bar{\qH}_j^H - |z|^2\qI_N\right) \qPsi^H\right) \right\}\notag\\
    =& \Im\left\{\frac{1}{N}\tr\left(\qR_i \qPsi \left( z\qS - \sum_{j=1}^{K} |z|^2\tilde{e}_j^*(z)\qR_j
    + \sum_{j=1}^{K} \bar{\qH}_j \tilde{\qPhi}_{zj} \left(-z \tilde{\qPhi}_{zj}^{-1}\right) \tilde{\qPhi}_{zj}^H \bar{\qH}_j^H - |z|^2\qI_N\right) \qPsi^H\right) \right\}
    \notag\\
    \eqover{(b)}& \Im\left\{\frac{1}{N}\tr\left(\qR_i \qPsi \left( z\qS - \sum_{j=1}^{K} |z|^2\tilde{e}_j^*(z)\qR_j
    + \sum_{j=1}^{K} \bar{\qH}_j \tilde{\qPhi}_{zj} (z\qI_{n_j} + \beta_j ze_j(z)\qT_j) \tilde{\qPhi}_{zj}^H \bar{\qH}_j^H - |z|^2\qI_N\right) \qPsi^H\right) \right\}
    \notag\\
    =& \sum_{j=1}^{K} \tilde{e}_{j,2}(z)  |z|^2 \frac{1}{N} \tr\left(\qR_i\qPsi \qR_j\qPsi^H\right)
     + \sum_{j=1}^{K} \Im\{ze_{j}(z)\} \frac{\beta_j}{N} \tr\left(\qR_i \qPsi \bar{\qH}_j \tilde{\qPhi}_{zj}\qT_j\tilde{\qPhi}_{zj}^H \bar{\qH}_j^H\qPsi^H \right)
    \notag\\
     & + \frac{z_2}{N}\tr \left(\qR_i \qPsi (\qS + \bar{\qH} \tilde{\qPhi}_z \tilde{\qPhi}_z^H \bar{\qH}^H)\qPsi^H\right)\notag \\
     \eqover{(c)}& \sum_{j=1}^{K} \tilde{e}_{j,2}(z)  |z|^2 u_{1,ij} + \sum_{j=1}^{K} \Im\{ze_{j}(z)\} u_{2,ij} + \frac{z_2}{N}\tr\left( \qR_i \qPsi \left(\qS + \bar{\qH} \tilde{\qPhi}_z \tilde{\qPhi}_z^H \bar{\qH}^H\right)\qPsi^H \right),
\end{align}
where (a) and (b) are obtained by expanding $\left(z\qPsi^{-H}\right)$ and $\left(-z \tilde{\qPhi}_{zj}^{-1}\right)$ respectively using
(\ref{eq:fixedPoint2Recall}), and (c) is obtained by using the definitions in (\ref{eq:defU}). Similarly we can get
\begin{align}
\Im\{z\tilde{e}_i(z)\}=& \sum_{j=1}^{K} e_{j,2}(z)   |z|^2 v_{1,ij} + \sum_{j=1}^{K} \Im\{z\tilde{e}_j(z)\} v_{2,ij}
    + \frac{z_2}{n_i} \tr\left( \underline{\qT}_i \tilde{\qPsi} \bar{\qH}^H \qPhi_z \qPhi_z^H \bar{\qH} \tilde{\qPsi}^H \right),\\
e_{i,2}(z)=& \sum_{j=1}^{K} \Im\{z\tilde{e}_j(z)\}  u_{1,ij} + \sum_{j=1}^{K} e_{j,2}(z) u_{2,ij} + \frac{z_2}{N}\tr\left( \qR_i\qPsi\qPsi^H \right),\\
\tilde{e}_{i,2}(z) =& \sum_{j=1}^{K} \Im\{ze_j(z)\}  v_{1,ij} + \sum_{j=1}^{K} \tilde{e}_{j,2}(z) v_{2,ij} + \frac{z_2}{n_i} \tr\left( \underline{\qT}_i\tilde{\qPsi}\left(\qI_n + \frac{1}{|z|^2} \bar{\qH}^H\qPhi_z\qS\qPhi_z^H\bar{\qH} \right) \tilde{\qPsi}^H \right).
\end{align}

Let
\begin{equation*}
\qeta = \left[ e_{1,2}(z), \dots, e_{K,2}(z), \Imag{ze_1(z)}, \dots, \Imag{ze_K(z)}, \tilde{e}_{1,2}(z), \dots \tilde{e}_{K,2}(z), \Imag{z\tilde{e}_i(z)}, \dots, \Imag{z\tilde{e}_K(z)} \right]^T.
\end{equation*}
By the definition of $\qGamma$ in (\ref{eq:defGamma}), $\qeta$ satisfies
\begin{equation} \label{eq:ineq_cb}
    \qeta = \qGamma \qeta + \qb,
\end{equation}
where $\qb = [\qb_1^{T} \, \qb_2^{T} \, \qb_3^{T} \, \qb_4^{T} ]^T$ with $\qb_1 =[\frac{ b_{1,i} }{1- u_{2,ii}} ], ~\qb_2 = [\frac{ b_{2,i} }{1- u_{2,ii}} ],
~\qb_3 = [\frac{ b_{3,i} }{1- v_{2,ii}} ], ~\qb_4 = [\frac{ b_{4,i} }{1- v_{2,ii}} ] \in \bbC^{K}$, and
\begin{equation}
\left\{\begin{aligned}
    b_{1,i} &= \frac{z_2}{N}\tr\left( \qR_i\qPsi\qPsi^H \right), \\
    b_{2,i} &= \frac{z_2}{N}\tr\left( \qR_i \qPsi \left(\qS + \bar{\qH} \tilde{\qPhi}_z \tilde{\qPhi}_z^H \bar{\qH}^H\right) \qPsi^H \right), \\
    b_{3,i} &= \frac{z_2}{n_i} \tr\left( \underline{\qT}_i\tilde{\qPsi}\left(\qI_n + \frac{1}{|z|^2} \bar{\qH}^H\qPhi_z^H\qS\qPhi_z\bar{\qH} \right) \tilde{\qPsi}^H \right), \\
    b_{4,i} &= \frac{z_2}{n_i} \tr\left( \underline{\qT}_i \tilde{\qPsi} \bar{\qH}^H \qPhi_z \qPhi_z^H \bar{\qH} \tilde{\qPsi}^H \right).
\end{aligned}\right.
\end{equation}
Let
\begin{equation}
    \qUpsilon = \diag \left(1- u_{2,11}, \cdots, 1- u_{2,KK}, 1- v_{2,11}, \cdots, 1-v_{2,KK} \right).
\end{equation}
Multiplying both sides of (\ref{eq:ineq_cb}) by $\qUpsilon$ gives
\begin{equation} \label{eq:mulIneq_cb}
    \qUpsilon\qeta = \qUpsilon\qGamma \qeta + \qUpsilon\qb.
\end{equation}
For $z \in \bbC^+$, it is observed that the following quantities
\begin{equation} \label{eq:posQua}
\left\{\begin{aligned}
    & e_{i,2}(z),~\Im\{ze_i(z)\},~\tilde{e}_{i,2}(z),~\Im\{z\tilde{e}_i(z)\}, ~\forall i, \\
    & u_{1,ij},~u_{2,ij},~v_{1,ij},~v_{2,ij}, ~\forall i,j, \\
    & b_{1,i},~b_{2,i},~b_{3,i},~b_{4,i}, ~\forall i.
\end{aligned}\right.
\end{equation}
are all positive. For any matrix $\qA = [a_{ij}]$, we write $\qA > \qzero$ if $a_{ij} > 0~\forall i,j$. From (\ref{eq:posQua}), we have that $\qeta >\qzero$,
$\qUpsilon\qGamma > \qzero$, and $\qUpsilon\qb > \qzero$. As a result, we get $\qUpsilon\qeta > \qzero$ [the right-hand side of (\ref{eq:mulIneq_cb})] and since
$\qeta >\qzero$, we conclude that
\begin{equation} \label{eq:a6}
1- u_{2,ii} > 0, ~~~~ 1- v_{2,ii} > 0, ~\forall i.
\end{equation}
Now, all the elements of $\qGamma$, $\qeta$, and $\qb$ are shown to be positive. Therefore, from (\ref{eq:ineq_cb}) and Lemma \ref{lineIneqLemma}, we get
$\rho(\qGamma) < 1$. Similarly, we also have $\rho(\qGamma^{\circ}) < 1$.

A standard computation involving the resolvent identity (Lemma \ref{lemma_Resolvent}) yields
\begin{align}
    & e_i(z) - e_i^{\circ}(z) \notag \\
    =& - \frac{1}{N}\tr\left(\qR_i \qPsi \left( \qPsi^{-1} - \qPsi^{\circ -1} \right) \qPsi^{\circ}\right) \notag \\
     \eqover{(a)}& \sum_{j=1}^{K} (z\tilde{e}_j(z)-z\tilde{e}^{\circ}_j(z)) \frac{1}{N} \tr\left(\qR_i\qPsi \qR_j \qPsi^{\circ}\right)
     + \frac{1}{N} \sum_{j=1}^{K} \tr\left( \qR_i\qPsi \bar{\qH}_j (\tilde{\qPhi}_{zj} - \tilde{\qPhi}_{zj}^{\circ}) \bar{\qH}_j^H \qPsi^{\circ}\right) \notag \\
    =& \sum_{j=1}^{K} (z\tilde{e}_j(z)-z\tilde{e}^{\circ}_j(z)) \frac{1}{N} \tr\left(\qR_i\qPsi \qR_j \qPsi^{\circ}\right)
     + \sum_{j=1}^{K} (e_{j}(z)-e^{\circ}_{j}(z)) \frac{\beta_j}{N} \tr\left( \qR_i\qPsi \bar{\qH}_j \tilde{\qPhi}_{zj} \qT_j \tilde{\qPhi}_{zj}^{\circ} \bar{\qH}_j^H \qPsi^{\circ}\right).
     \label{eq:difpis}
\end{align}
where (a) is obtained by expanding $\left( \qPsi^{-1} - \qPsi^{\circ -1} \right)$ using (\ref{eq:fixedPoint2Recall}). Similarly,
\begin{equation} \label{eq:difgamma}
\begin{aligned}
    &\tilde{e}_i(z) - \tilde{e}_i^{\circ}(z) \\
    =&
    \sum_{j=1}^{K} (ze_j(z)-ze^{\circ}_j(z)) \frac{\beta_j}{n_j} \tr\left( \underline{\qT}_i \tilde{\qPsi} \underline{\qT}_j \tilde{\qPsi}^{\circ}\right)
    +
    \sum_{j=1}^{K}  (\tilde{e}_j(z)-\tilde{e}^{\circ}_j(z))  \frac{1}{n_j}
    \tr \left(\underline{\qT}_i \tilde{\qPsi} \bar{\qH}^H \qPhi_z \qR_j \qPhi_z^{\circ} \bar{\qH} \tilde{\qPsi}^{\circ}\right).
\end{aligned}
\end{equation}

Now, let $\qtau \triangleq [\qe^T ~\qe_z^T~\tilde\qe^T~\tilde\qe_z^T]^T$ and $\qtau^{\circ} \triangleq [\qe^{\circ T} ~\qe_z^{\circ T}~\tilde\qe^{\circ
T}~\tilde\qe_z^{\circ T}]^T$. Thus we have
\begin{equation} \label{eq:LDifSysEta}
    \qtau-\qtau^{\circ} = \qDelta (\qtau-\qtau^{\circ}),
\end{equation}
where
\begin{equation*}
    \qDelta = \left[
          \begin{array}{cccc}
          \qDelta_{11} & \qzero & \qzero & \qDelta_{12}\\
          \qzero & \qDelta_{11} & z^2\qDelta_{12} & \qzero\\
          \qzero & \qDelta_{21} & \qDelta_{22} & \qzero\\
          z^2\qDelta_{21} & \qzero & \qzero & \qDelta_{22}\\
          \end{array}
          \right]
\end{equation*}
with $\qDelta_{11} = [\Delta_{11,ij}] \in \bbC^{K \times K}, ~\qDelta_{12} = [\Delta_{12,ij}] \in \bbC^{K \times K}, ~\qDelta_{21} = [\Delta_{21,ij}] \in \bbC^{K
\times K}, ~\qDelta_{22} = [\Delta_{22,ij}] \in \bbC^{K \times K}$ and
\begin{align*}
    \Delta_{11,ij} &=
    \left\{\begin{aligned}
    &0, & & i=j, \\
    &\frac{ \frac{\beta_j}{N} \tr\left( \qR_i \qPsi \bar{\qH}_j \tilde{\qPhi}_{zj} \qT_j \tilde{\qPhi}_{zj}^{\circ} \bar{\qH}_j^H \qPsi^{\circ} \right) }
    {1 - \frac{\beta_j}{N} \tr\left( \qR_i \qPsi \bar{\qH}_i \tilde{\qPhi}_{zi} \qT_i \tilde{\qPhi}_{zi}^{\circ} \bar{\qH}_i^H \qPsi^{\circ} \right)},
    & & i \neq j,
    \end{aligned}\right. \\
    \Delta_{12,ij} &=
    \frac{ \frac{1}{N} \tr\left(\qR_i \qPsi \qR_j \qPsi^{\circ} \right) }
    {1 - \frac{\beta_j}{N} \tr\left( \qR_i\qPsi \bar{\qH}_i \tilde{\qPhi}_{zi} \qT_i \tilde{\qPhi}_{zi}^{\circ} \bar{\qH}_i^H \qPsi^{\circ}\right)}, \\
    \Delta_{21,ij} &=
    \frac{ \frac{\beta_j}{n_j} \tr\left( \underline{\qT}_i \tilde{\qPsi} \underline{\qT}_j \tilde{\qPsi}^{\circ}\right) }
    {1- \frac{1}{n_j} \tr\left( \underline{\qT}_i \tilde{\qPsi}_z \bar{\qH}^H \qPhi_z \qR_i \qPhi^{\circ} \bar{\qH} \tilde{\qPsi}^{\circ} \right)},\\
    \Delta_{22,ij} &=\left\{\begin{aligned}
    &0, & & i=j, \\
    &\frac{ \frac{1}{n_j} \tr\left( \underline{\qT}_i \tilde{\qPsi} \bar{\qH}^H \qPhi_z \qR_j \qPhi_z^{\circ} \bar{\qH} \tilde{\qPsi}^{\circ}\right) }
    {1- \frac{1}{n_j} \tr \left( \underline{\qT}_i \tilde{\qPsi}_z \bar{\qH}^H \qPhi_z \qR_i \qPhi^{\circ} \bar{\qH} \tilde{\qPsi}^{\circ} \right)},
    & & i \neq j.
    \end{aligned}\right.
\end{align*}

Let $|\cdot|$ denote the operator taking the absolute values of the input vector or matrix. It follows from Lemma \ref{matrixIneqLemma1} that $\rho(\qDelta) \leq
\rho(|\qDelta|)$. Applying Lemma \ref{matrixIneqLemma3} with $\qA = \sqrt{\beta_j/N} \, \qR_i^\frac{1}{2} \qPsi \bar{\qH}_i \tilde{\qPhi}_{zi} \qT_i^\frac{1}{2}$
and $\qB^H = \sqrt{\beta_j/N}\qT_i^\frac{1}{2} \tilde{\qPhi}_{zi}^{\circ} \bar{\qH}_i^H \qPsi^{\circ} \qR_i^\frac{1}{2} $, we have a lower bound for the
denominator of $\Delta_{11,ij}$ by
\begin{align}
    &\left(1 - \frac{\beta_j}{N} \tr\left( \qR_i\qPsi \bar{\qH}_i \tilde{\qPhi}_{zi} \qT_i \tilde{\qPhi}_{zi}^{\circ} \bar{\qH}_i^H \qPsi^{\circ} \right)\right) \notag \\
    \geq&\left(1 - \frac{\beta_j}{N} \tr\left( \qR_i\qPsi \bar{\qH}_i \tilde{\qPhi}_{zi} \qT_i \tilde{\qPhi}_{zi}^H \bar{\qH}_i^H \qPsi^H \right)\right)^{\frac{1}{2}}
    \left(1 - \frac{\beta_j}{N} \tr \left( \qR_i \qPsi^{\circ H} \bar{\qH}_i \tilde{\qPhi}_{zi}^{\circ H} \qT_i \tilde{\qPhi}_{zi}^{\circ} \bar{\qH}_i^H \qPsi^{\circ} \right)\right)^{\frac{1}{2}}\label{eq:denomBound}
\end{align}
where the conditions $\tr(\qA\qA^H) = u_{2,ii}\leq 1$ and $\tr(\qB\qB^H) = u_{2,ii}^{\circ} \leq 1$ are satisfied by (\ref{eq:a6}). Applying the Cauchy-Schwarz
inequality to the numerator of $\Delta_{11,ij}$, we then obtain from (\ref{eq:denomBound})
\begin{multline} \label{eq:IneqDelta11}
    |\Delta_{11,ij}|
    \leq
    \left(
    \frac{ \frac{\beta_j}{N} \tr\left( \qR_i\qPsi \bar{\qH}_j \tilde{\qPhi}_{zj} \qT_j \tilde{\qPhi}_{zj}^H \bar{\qH}_j^H \qPsi^H \right) }
    {1 - \frac{\beta_j}{N} \tr\left( \qR_i\qPsi \bar{\qH}_i \tilde{\qPhi}_{zi} \qT_i \tilde{\qPhi}_{zi}^H \bar{\qH}_i^H \qPsi^H \right)}
    \right)^{\frac{1}{2}}
    \left(
    \frac{ \frac{\beta_j}{N} \tr\left(\qR_i \qPsi^{\circ H} \bar{\qH}_j \tilde{\qPhi}_{zj}^{\circ H} \qT_j \tilde{\qPhi}_{zj}^{\circ} \bar{\qH}_j^H \qPsi^{\circ} \right) }
    {1 - \frac{\beta_j}{N} \tr\left( \qR_i\qPsi^{\circ H} \bar{\qH}_i \tilde{\qPhi}_{zi}^{\circ H} \qT_i \tilde{\qPhi}_{zi}^{\circ} \bar{\qH}_i^H \qPsi^{\circ}\right)}
    \right)^{\frac{1}{2}}.
\end{multline}
Recalling the definitions of the entries of $\qGamma$, (\ref{eq:IneqDelta11}) is equivalent to
\begin{align}
|\Delta_{11,ij}|&\leq \left|\frac{ u_{2,ij} }{1- u_{2,ii}}\right|^\frac{1}{2} \left|\frac{ u_{2,ij}^{\circ} }{1- u_{2,ii}^{\circ}}\right|^\frac{1}{2}= |\Gamma_{11,ij}|^\frac{1}{2} |\Gamma_{11,ij}^{\circ}|^\frac{1}{2}.
\end{align}
Likewise we have
\begin{equation}
    |\Delta_{12,ij}|  \leq |\Gamma_{12,ij}|^\frac{1}{2} |\Gamma_{12,ij}^{\circ}|^\frac{1}{2},~~
    |\Delta_{21,ij}|  \leq |\Gamma_{21,ij}|^\frac{1}{2} |\Gamma_{21,ij}^{\circ}|^\frac{1}{2}, ~~\mbox{and}~~
    |\Delta_{22,ij}|  \leq |\Gamma_{22,ij}|^\frac{1}{2} |\Gamma_{22,ij}^{\circ}|^\frac{1}{2}.
\end{equation}
We then conclude from Lemmas \ref{matrixIneqLemma1} and \ref{matrixIneqLemma2} that
\begin{equation}
    \rho(|\qDelta|) \leq \rho\left( \left[ |\Gamma_{ij}|^{\frac{1}{2}} |\Gamma_{ij}^{\circ}|^{\frac{1}{2}} \right] \right)
    \leq \rho(\qGamma)^{\frac{1}{2}} \rho(\qGamma^{\circ})^{\frac{1}{2}} < 1
\end{equation}
where the fact that $\rho(\qGamma) < 1$ and $\rho(\qGamma^{\circ}) < 1$ is proved before. As pointed out by \cite{Cou-09}, this contradicts to the statement that
$\qDelta$ has an eigenvalue equal to $1$ via (\ref{eq:LDifSysEta}). Therefore we get $\qe=\qe^{\circ}$ and $\tilde{\qe} = \tilde{\qe}^{\circ}$ if $z \in \bbC^+$.
If $z \in \bbC^{-}$ or $z \in \bbR^{-}$, similar arguments apply and details are omitted here. Theorem \ref{mainTh_uniq} is thus proved.

\section*{\sc Appendix C. Proof of Theorem \ref{mainTh_Cap}}
Recalling (\ref{eq:shannonTrans}), we have \cite[page 891]{Hac-07}
\begin{equation} \label{eq:a7}
     \calV_{\qB_N}(\sigma^2) = \int_{\sigma^2}^{\infty} \left( \frac{1}{\omega} - m_{\qB_N}(-\omega) \right) d\omega, ~~\mbox{for }\sigma^2 \in \bbR^+.
\end{equation}
In Appendix C.1, we first show $\Ex\{\calV_{\qB_N}(\sigma^2)\} - \calV_N(\sigma^2) \rightarrow 0$; i.e.,
\begin{equation} \label{eq:capConverg}
    \int_{\sigma^2}^{\infty} \left( \frac{1}{\omega} - \Ex\{ m_{\qB_N}(-\omega) \} \right) d\omega
    - \int_{\sigma^2}^{\infty} \left( \frac{1}{\omega} - \frac{1}{N} \tr\left(\qPsi(-\omega)\right) \right) d\omega
    \xrightarrow{\largeN\rightarrow \infty} 0.
\end{equation}
Next, we show in Appendix C.2 that with the additional assumptions in Remark 2, (\ref{eq:capConverg}) can be strengthened to almost surely convergence as
$\largeN\rightarrow \infty$. Finally, in Appendix C.3, we show that $\int_{\sigma^2}^{\infty} \frac{1}{\omega} - \frac{1}{N} \tr\left(\qPsi(-\omega)\right)
d\omega$ can be written more explicitly as (\ref{eq:AsyShannon}).

\subsection*{\sc Appendix C.1 Proof of the convergence of $\Ex\{\calV_{\qB_N}(\sigma^2)\} - \calV_N(\sigma^2)$}
By Theorem \ref{mainTh_Stj} together with the dominated convergence theorem, we have
\begin{align}
    \left( \frac{1}{\omega} - \Ex\{ m_{\qB_N}(-\omega) \} \right) - \left( \frac{1}{\omega} - \frac{1}{N} \tr\left(\qPsi(-\omega)\right) \right) ~\xrightarrow{\largeN\rightarrow \infty} 0 .
\end{align}
Let $F_N$ be the probability distribution whose Stieltjes transform is $\frac{1}{N} \tr\qPsi(z)$. Notice that \cite{Hac-07}
\begin{align}
    &\left| \left( \frac{1}{\omega} - \Ex\{ m_{\qB_N}(-\omega) \} \right) - \left( \frac{1}{\omega} - \frac{1}{N} \tr\left(\qPsi(-\omega)\right) \right) \right| \notag \\
    &~~~~~~~~~~\leq \left| \Ex \left\{ \int_{0}^{\infty} \left( \frac{1}{\omega} - \frac{1}{\lambda+\omega} \right) d F_{\qB_N} (\lambda) \right\} \right|
    + \left| \int_{0}^{\infty} \left( \frac{1}{\omega} - \frac{1}{\lambda+\omega} \right) d F_N (\lambda) \right| \notag \\
    &~~~~~~~~~~\leq \left| \frac{1}{\omega^2} \Ex \left\{ \int_{0}^{\infty} \lambda d F_{\qB_N} (\lambda) \right\} \right|
    + \left| \frac{1}{\omega^2} \int_{0}^{\infty}  \lambda d F_N (\lambda) \right|. \label{eq:a8}
\end{align}
Also, we notice the following equalities:
\begin{align}
    &\Ex\left\{ \int_{0}^{\infty}\lambda d F_{\qB_N}(\lambda) \right\} = \frac{1}{N} \sum_{k=1}^{K} \frac{\tr(\qT_k) \tr(\qR_k)}{n_k} + \frac{1}{N} \sum_{k=1}^{K} \tr(\bar\qH_k \bar\qH_k^H), \label{eq:eqTrue}\\
    &\int_{0}^{\infty}\lambda d F_N(\lambda) = \frac{1}{N} \sum_{k=1}^{K} \frac{\tr(\qT_k) \tr(\qR_k)}{n_k} + \frac{1}{N} \sum_{k=1}^{K} \tr(\bar\qH_k \bar\qH_k^H). \label{eq:eqAsy}
\end{align}
We will confirm these equalities later. From (\ref{eq:eqTrue}) and (\ref{eq:eqAsy}) together with the constraints in (\ref{eq:normAss}), we get
\begin{equation} \label{eq:estBound}
    \Ex\left\{ \int_{0}^{\infty}\lambda d F_{\qB_N}(\lambda) \right\}  = 2K ~~~ \mbox{and} ~~~
    \int_{0}^{\infty}\lambda d F_N(\lambda) = 2K.
\end{equation}
As a result, (\ref{eq:capConverg}) follows from the dominated convergence theorem.

It remains to check (\ref{eq:eqTrue}) and (\ref{eq:eqAsy}). For (\ref{eq:eqTrue}), a direct calculation yields
\begin{align}
    \Ex\left\{ \int_{0}^{\infty}\lambda d F_{\qB_N}(\lambda) \right\}
    =\frac{1}{N} \Ex \left\{ \tr(\qH\qH^H) \right\}
    =& \frac{1}{N} \sum_{k=1}^{K} \Ex \left\{  \tr(\qR_k\qX_k\qT_k\qX_k^H) \right\} + \frac{1}{N} \sum_{k=1}^{K} \tr\left( \bar\qH_k\bar\qH_k^H \right)  \label{eq:eqTrue2} \\
    =& \frac{1}{N} \sum_{k=1}^{K} \frac{\tr(\qT_k)\tr(\qR_k)}{n_k} + \frac{1}{N} \sum_{k=1}^{K} \tr\left( \bar\qH_k\bar\qH_k^H \right). \label{eq:powOfCh}
\end{align}
The equality (\ref{eq:eqAsy}) can be obtained by using \cite[(C.4)]{Hac-07}:
\begin{equation}
    \int_{0}^{\infty}\lambda d F_N(\lambda) = \lim_{z_2 \rightarrow \infty} {\rm Re}\left\{ -\sfj z_2 \left( \sfj z_2 \frac{1}{N} \tr\left(\qPsi(\sfj z_2)\right) + 1 \right)\right\}. \label{eq:eqAsy2}
\end{equation}
The proof of (\ref{eq:eqAsy2}) being the right-hand slide of (\ref{eq:eqAsy}) is similar to that in \cite[Lemma C.1]{Hac-07} and is therefore omitted. The proof of (\ref{eq:a4}) is complete.

\subsection*{\sc Appendix C.2 Proof of $\calV_{\qB_N}(\sigma^2) - \calV_N(\sigma^2) \xrightarrow{a.s.} 0 $}
Following \cite{Cou-09}, one can verify (\ref{eq:a5}) under the assumptions of Theorem \ref{AsConv_Cap}.

As for Remark 2, similar to (\ref{eq:eqTrue2}),  write
\begin{equation}
\int_{0}^{\infty}\lambda d F_{\qB_N}(\lambda)=\frac{1}{N} \sum_{k=1}^{K}\tr(\qR_k\qX_k\qT_k\qX_k^H)
+\frac{2}{N} \sum_{k=1}^{K} \tr \left( \qR_k^{\frac{1}{2}}\qX_k\qT_k^{\frac{1}{2}}\bar\qH_k ^H \right)+\frac{1}{N} \sum_{k=1}^{K} \tr(\bar\qH_k \bar\qH_k^H).
\end{equation}
Furthermore, write
\begin{equation}
\frac{1}{N} \sum_{k=1}^{K}\tr(\qR_k\qX_k\qT_k\qX_k^H) = \frac{1}{N}\sum_{k=1}^{K}\sum_{i=1}^N\sum_{j=1}^NR_{ij}^{(k)}\acute{\qx}_{k,j}^H\qT_k\acute{\qx}_{k,i},
\end{equation}
and
\begin{equation}
\frac{1}{N} \sum_{k=1}^{K} \tr( \qR_k^{\frac{1}{2}}\qX_k\qT_k^{\frac{1}{2}}\bar\qH_k ^H)=\frac{1}{N} \sum_{k=1}^{K}\sum_{j=1}^N\acute{\qx}_{k,j}^H\qy_{k,j}
\end{equation}
where $\qR_k=[R_{ij}^{(k)}]$ and $\acute{\qx}_{k,j}^H$ is the $j$th row of $\qX_k$ and $\qy_{k,j}=\qT_k^{\frac{1}{2}}\bar\qH_k ^H\qR_k^{\frac{1}{2}}\qe_j$. Note
that
\begin{equation}
\frac{1}{n_k}\sum_{i=1}^{N} \left|R_{ii}^{(k)}\right|^2\leq \frac{1}{n_k}\tr(\qR_k^2).
\end{equation}
It follows from Lemma \ref{sumIneq}, Lemma \ref{Burk2} and Lemma \ref{traceIneq} that
\begin{align}
\Ex\left\{ \left|\frac{1}{N}\sum_{k=1}^K\sum_{i=1}^N R_{ii}^{(k)} \left(\acute{\qx}_{k,i}^H\qT_k\acute{\qx}_{k,i}-\frac{1}{n_k}\tr(\qT_k)\right) \right|^2 \right\}
& \leq \frac{C}{N^2}\sum_{k=1}^K \sum_{i=1}^N |R_{ii}^{(k)}|^2\Ex\left\{\left|\acute{\qx}_{k,i}^H\qT_k\acute{\qx}_{k,i}-\frac{1}{n_k}\tr(\qT_k)\right|^2 \right\}\notag \\
& \leq \frac{C'}{N^2}\sum_{k=1}^K\sum_{i=1}^N |R_{ii}^{(k)}|^2\frac{1}{n_k^2}\tr(\qT_k^2) \notag \\
& \leq C'' \sum_{k=1}^K\frac{1}{N n_k}\frac{1}{N}\tr(\qR_k^2)\frac{1}{n_k}\tr(\qT_k^2).
\end{align}
This, together with Borel-Cantelli's Lemma, implies
\begin{equation}
\frac{1}{N}\sum_{k=1}^K\sum_{i=1}^N R_{ii}^{(k)}\acute{\qx}_{k,i}^H\qT_k\acute{\qx}_{k,i}-\frac{1}{N} \sum_{k=1}^{K} \frac{\tr(\qT_k) \tr(\qR_k)}{n_k}
 \xrightarrow{a.s.} 0.
\end{equation}
A direct calculation indicates that
\begin{align}
\Ex \left\{ \left|\frac{1}{N}\sum_{k=1}^K\sum_{i\neq j}^NR_{ij}^{(k)}\acute{\qx}_{k,j}^H\qT_k\acute{\qx}_{k,i}\right|^2 \right\}
& \leq \frac{C}{N^2}\sum_{k=1}^K\sum_{i\neq j}^N|R_{ij}^{(k)}|^2\frac{1}{n_k^2}\tr(\qT_k^2) \notag \\
& \leq C'\sum_{k=1}^K\frac{1}{N n_k}\frac{1}{N}\tr(\qR_k^2)\frac{1}{n_k}\tr(\qT_k^2).
\end{align}
It follows that
\begin{equation}\label{eq:a10}
\frac{1}{N} \sum_{k=1}^{K}\tr(\qR_k\qX_k\qT_k\qX_k^H)-\frac{1}{N} \sum_{k=1}^{K} \frac{\tr(\qT_k) \tr(\qR_k)}{n_k} \xrightarrow{a.s.} 0.
\end{equation}
Similarly, we have
\begin{align}
\Ex \left\{ \left|\frac{1}{N} \sum_{k=1}^{K}\sum_{j=1}^N\acute{\qx}_{k,j}^H\qy_{k,j}\right|^2 \right\}
&\leq C \sum_{k=1}^K\frac{1}{N n_k} \frac{1}{N} \tr(\qT_k\bar{\qH}_k^H\qR_k\bar{\qH}_k) \notag \\
&\leq C' \sum_{k=1}^K \frac{1}{N n_k} \left(\frac{1}{N}\tr(\qT_k^4)\frac{1}{N}\tr(\qR_k^4) \left( \frac{1}{N}\tr\left((\bar{\qH}_k^H\bar{\qH}_k)^2\right)\right)^2 \right)^{\frac{1}{4}},
\end{align}
which implies that
\begin{equation}\label{eq:a11}
\frac{1}{N} \sum_{k=1}^{K} \tr( \qR_k^{\frac{1}{2}}\qX_k\qT_k^{\frac{1}{2}}\bar\qH_k ^H)\stackrel{a.s.}\longrightarrow0.
\end{equation}
It follows from the generalized dominated convergence theorem, (\ref{eq:a10}), (\ref{eq:a11}), (\ref{eq:a7}) and (\ref{eq:a8}) that Remark 2 is true.

\subsection*{\sc Appendix C.3 Explicit Expression of $\int_{\sigma^2}^{\infty} \left( \frac{1}{\omega} - \frac{1}{N} \tr\left(\qPsi(-\omega)\right) \right) d\omega$}
In this appendix, we will prove
\begin{equation}
 \calV_N(\sigma^2) = \int_{\sigma^2}^{\infty} \left( \frac{1}{\omega} - \frac{1}{N} \tr\left(\qPsi(-\omega)\right) \right) d\omega,
\end{equation}
or equivalently,
\begin{equation} \label{eq:partialEta}
 \frac{\partial \calV_N(\sigma^2)}{\partial \sigma^2}= \frac{1}{N} \tr\left(\qPsi(-\sigma^2)\right) - \frac{1}{\sigma^2}.
\end{equation}
The right-hand side of (\ref{eq:partialEta}) can be reexpressed as
\begin{align}
\frac{1}{N} \tr\left(\qPsi(-\sigma^2)\right) - \frac{1}{\sigma^2}
 =& \frac{1}{N} \tr\left( \qPsi(-\sigma^2) -  (\sigma^2\qI)^{-1} \right) \notag \\
 \eqover{(a)}& - \frac{1}{N} \tr\left(\qPsi(-\sigma^2) \left( \sum_{i=1}^{K} \sigma^2 \tilde{e}_i(-\sigma^2) \qR_i  + \bar{\qH}(\sigma^2 \tilde{\qPhi}(-\sigma^2))\bar{\qH}^H \right) (\sigma^2\qI)^{-1}\right) \notag \\
 \eqover{(b)}& - \sum_{i=1}^{K} \tilde{e}_i(-\sigma^2) e_i(-\sigma^2) - \frac{1}{N} \tr\left(\qPsi(-\sigma^2)\bar{\qH}\tilde{\qPhi}(-\sigma^2)\bar{\qH}^H\right), \label{eq:difEtaR}
\end{align}
where (a) is due to the resolvent identity (Lemma \ref{lemma_Resolvent}) and (b) follows merely from the definitions of $\tilde{e}_i(-\sigma^2)$ and
$e_i(-\sigma^2)$. We then prove that (\ref{eq:difEtaR}) corresponds to the left-hand side of (\ref{eq:partialEta}). To this end, we define
\begin{multline}
    \calV(\sigma^2,\qkappa,\tilde{\qkappa}) \triangleq \frac{1}{N} \log\det\left( \qI_N + \sum_{i=1}^{K} \tilde{\kappa}_i \qR_i
    + \frac{1}{\sigma^2} \sum_{i=1}^{K} \bar{\qH}_i (\qI_{n_i}+ \beta_i \kappa_i \qT_i)^{-1} \bar{\qH}_i^H \right) +\\
    \frac{1}{N} \sum_{i=1}^{K} \log\det\left( \qI_{n_i} + \beta_i \kappa_i \qT_i \right)
                           - \sigma^2 \sum_{i=1}^{K} \kappa_i \tilde{\kappa}_i.
\end{multline}
Note that $\calV(\sigma^2, \qe_{\sigma}, \tilde{\qe}_{\sigma}) = \calV_N(\sigma^2)$ with $\qe_{\sigma} \triangleq [e_i(-\sigma^2)] \in \bbR^K$ and
$\tilde{\qe}_{\sigma} \triangleq [\tilde{e}_i(-\sigma^2)] \in \bbR^K$. The derivative of $\calV_N(\sigma^2)$ can be expressed as
\begin{equation}
    \frac{\partial \calV_N(\sigma^2)}{\partial \sigma^2}
    = \left.\frac{\partial \calV}{\partial \sigma^2}\right|_{(\sigma^2, \qe_{\sigma}, \tilde{\qe}_{\sigma})}
    + \sum_{i=1}^{K} \left.\frac{\partial \calV}{\partial \kappa_i}\right|_{(\sigma^2, \qe_{\sigma}, \tilde{\qe}_{\sigma})} \frac{\partial e_i}{\partial \sigma^2}
    + \sum_{i=1}^{K} \left.\frac{\partial \calV}{\partial \tilde{\kappa}_i}\right|_{(\sigma^2, \qe_{\sigma}, \tilde{\qe}_{\sigma})} \frac{\partial \tilde{e}_i}{\partial \sigma^2}.
\end{equation}
It can be checked that
\begin{equation}
    \left.\frac{\partial \calV}{\partial \kappa_i}\right|_{(\sigma^2, \qe_{\sigma}, \tilde{\qe}_{\sigma})} = 0,~~\mbox{and}~~
    \left.\frac{\partial \calV}{\partial \tilde{\kappa}_i}\right|_{(\sigma^2, \qe_{\sigma}, \tilde{\qe}_{\sigma})} = 0, ~~\mbox{for }i =1,\dots,K.
\end{equation}
Therefore, we have
\begin{align}
    \frac{\partial \calV_N(\sigma^2)}{\partial \sigma^2}
    =& \left.\frac{\partial \calV}{\partial \sigma^2}\right|_{(\sigma^2, \qe_{\sigma}, \tilde{\qe}_{\sigma})}\notag \\
    =&  - \frac{1}{N} \tr \left[ \left( \qI_N + \sum_{i=1}^{K} \tilde{e}_i(-\sigma^2) \qR_i
    + \frac{1}{\sigma^2} \sum_{j=1}^{K} \bar{\qH}_j (\qI_{n_j} + \beta_j e_j(-\sigma^2) \qT_j)^{-1} \bar{\qH}_j^H \right)^{-1} \right.\notag \\
    &~~~~~~~~~~~~~~~ \times \left.
    \left( \frac{1}{\sigma^4} \sum_{i=1}^{K} \bar{\qH}_i (\qI_{n_i} + \beta_i e_i(-\sigma^2) \qT_i)^{-1} \bar{\qH}_i^H \right) \right]
    - \sum_{i=1}^{K} e_i(-\sigma^2) \tilde{e}_i(-\sigma^2)\notag \\
    &= -\frac{1}{N} \tr\left(\qPsi(-\sigma^2)\bar{\qH}\tilde{\qPhi}(-\sigma^2)\bar{\qH}^H\right) - \sum_{i=1}^{K} \tilde{e}_i(-\sigma^2) e_i(-\sigma^2),
\end{align}
which is identical to (\ref{eq:difEtaR}) and hence we complete the proof.

\section*{\sc Appendix D. Mathematical Tools}
In this appendix, we provide some mathematical tools needed in the proof of Appendices A--C.

\begin{Lemma} {\rm \cite{Hor-91}} \label{LemmaNorm}
\begin{itemize}

\item[1)] Let $\qA = [A_{ij}]$ and $\qB$ be any matrices such that the product $\qA\qB$ exists and is a square matrix. Then
\begin{itemize}
\item[i)] $|\tr(\qA\qB)| \leq \| \qA \|_{\sfF} \| \qB \|_{\sfF}$,
\item[ii)] $\| \qA \qB \|_{\sfF} \leq \| \qA \|_{\sfF} \| \qB \|$,
\item[iii)] $\| \qA \qB \|_{\sfF} \leq \| \qA \|_{\sfF} \| \qB \|_{\sfF}$,
\item[iv)] $|A_{ij}| \leq \| \qA \|$.
\end{itemize}
\item[2)] If $\qA$ is nonnegative definite, we have $|\tr(\qA\qB)| \leq \| \qB \| \tr(\qA)$.
\item[3)] Let $\qA$ be any matrix such that the product $\qA\qB$ exists. Then, $ \|\qA \qB \| \leq \| \qA \| \| \qB \|$.
\end{itemize}
\end{Lemma}

\begin{Lemma} \label{lemma_Resolvent}
{\rm (Resolvent Identity \cite{Hac-07}.)} For invertible $\qA$ and $\qB$ matrices, we have the identity
\[
    \qA^{-1} - \qB^{-1} = -\qA^{-1}(\qA-\qB)\qB^{-1}.
\]
\end{Lemma}

\begin{Lemma} \label{rankIneq} {\rm (\cite[0.4.5 and 0.4.6]{Hor-85}).} Some fundamental equality and inequalities involving the rank are:

\begin{itemize}
\item[i)] If $\qA \in \bbC^{N \times n}$, $\rank(\qA) = \rank(\qA^H)$.
\item[ii)]  If $\qA \in \bbC^{N \times n}$ and $\qB \in \bbC^{n \times k}$, $\rank(\qA\qB) \leq \min\{ \rank(\qA),\rank(\qB)\}$.
\item[iii)] If $\qA, ~\qB \in \bbC^{N \times n}$, $\rank(\qA+\qB) \leq \rank(\qA) + \rank(\qB)$.
\item[iv)] If $\qA, ~\qB \in \bbC^{N \times n}$, $\rank([\qA~\qB]) \leq \rank(\qA) + \rank(\qB)$.
\end{itemize}
\end{Lemma}

\begin{Lemma} \label{disIneq}
{\rm (\cite[Theorem A.44]{Bai-10}.)} Let $\qA_1$ and $\qA_2$ be two $N \times n$ matrices. If $\qS$ and $\qD$ be Hermitian matrices of orders $N \times N$ and $n
\times n$, then we have
\[
    \sup_{x} \left|F_{\qS+\qA_1\qD\qA_1^H}(x) - F_{\qS+\qA_2\qD\qA_2^H}(x)\right|
    \leq \frac{1}{N} \rank(\qA_1-\qA_2).
\]
\end{Lemma}

\begin{Lemma} \label{disIneq2}
Let $\qS$ be a Hermitian matrix of order $N \times N$, $\qA_1$ and $\qA_2$ be $N \times n$ complex matrices, and $\qB_1$, $\qB_2$, $\qC_1$, $\qC_2$, and $\qD$ be
any matrices such that $\qB_1 \qD \qC_1$ and $\qB_2 \qD \qC_2$ exist and are of orders $N \times n$. Then,
\begin{multline*}
\sup_{x} \left|F_{\qS + (\qA_1 + \qB_1 \qD \qC_1)(\qA_1 + \qB_1 \qD \qC_1)^H}(x) - F_{\qS + (\qA_2 + \qB_2 \qD \qC_2)(\qA_2 + \qB_2 \qD \qC_2)^H}(x)\right| \\
    \leq\frac{1}{N} \left( \rank\left( \qA_1-\qA_2 \right) +  \rank\left( \qB_1-\qB_2 \right) + \rank\left( \qC_1-\qC_2 \right) \right) .
\end{multline*}
\end{Lemma}

\begin{proof}
\begin{align*}
    & \sup_{x} \left|F_{\qS + (\qA_1 + \qB_1 \qD \qC_1)(\qA_1 + \qB_1 \qD \qC_1)^H}(x) - F_{\qS + (\qA_2 + \qB_2 \qD \qC_2)(\qA_2 + \qB_2 \qD \qC_2)^H}(x)\right| \\
    \leqover{(a)} & \frac{1}{N} \rank\left( (\qA_1 + \qB_1 \qD \qC_1) - (\qA_2 + \qB_2 \qD \qC_2) \right) \\
    \leqover{(b)} & \frac{1}{N} \left( \rank\left( \qA_1-\qA_2 \right) +  \rank\left( \qB_1 \qD \qC_1-\qB_2 \qD \qC_2 \right) \right) \\
    = & \frac{1}{N} \left( \rank\left( \qA_1-\qA_2 \right) +  \rank\left( (\qB_1 \qD \qC_1 - \qB_2 \qD \qC_1) + (\qB_2 \qD \qC_1 -\qB_2 \qD \qC_2) \right) \right) \\
    \leq & \frac{1}{N} \left( \rank\left( \qA_1-\qA_2 \right) +  \rank\left( \qB_1 -\qB_2\right) +  \rank\left( \qC_1 -\qC_2\right) \right),
\end{align*}
where (a) follows from Lemma \ref{disIneq}, (b) follows from iii) of Lemma \ref{rankIneq}, and the last inequality follow from ii) and iii) of Lemma
\ref{rankIneq}.
\end{proof}

\begin{Lemma} \label{sumIneq}
For any $p \geq 1$ and real numbers $a_i$'s, we have
\begin{equation*}
    \left| \sum_{i=1}^n a_i \right|^p \leq n^{p-1} \sum_{i=1}^n |a_i|^p.
\end{equation*}
\end{Lemma}
\begin{proof}
This lemma follows from a simple application of the H\"older's inequality.
\end{proof}

\begin{Lemma} \label{lemma_eleIneq}
{\rm (Elementary Inequality \cite[page 29]{Bai-10}.)} If the $X_i$'s are independent with zero means, then
\begin{equation*}
    \Ex \left\{ \left| \sum X_i \right|^{2p}\right\} \leq C_p \left( \sum \Ex\left\{ |X_i|^{2p}\right\} + \left(\sum \Ex\left\{ |X_i|^{2}\right\}\right)^p \right).
\end{equation*}
\end{Lemma}

\begin{Lemma} \label{Burk2}
{\rm (Burkholder's Inequality \cite{Bur-73} or \cite[Lemma 2.12]{Bai-10}.)} Let $\{ X_i \}$ be a complex martingale difference sequence with respect to the
increasing $\sigma$-field $\{ {\cal F}_i\}$. Then for $p > 1$, we have
\begin{equation*}
    \Ex\left\{  \left| \sum X_i \right|^{p}\right\}\leq C_p \Ex\left\{ \left( \sum  |X_i|^{2}  \right)^\frac{p}{2}\right\}.
\end{equation*}
\end{Lemma}

\begin{Lemma} \label{PansIneq}
Let $\qx = [ \frac{1}{\sqrt{N}}X_i ] \in \bbC^{N}$ be a random vector, where $X_i$'s are independent complex random variables with zero mean and unit variance; and
$\qc = [c_i] \in \bbC^{N}$ be a deterministic vector independent of $\qx$. Assume that $|X_i|$'s are bounded by $\varepsilon\sqrt{N}$ with a constant $\varepsilon$
and $\|\qc\|^{2p}$ are bounded by a constant $C$ for $p \geq 1$. Then, for any $p \geq 1$, we have
\begin{equation}\label{eq:est_xc2p}
    \Ex \left\{| \qc^H \qx |^{2p}\right\} = O\left(\frac{1}{N}\right),
\end{equation}
and for $p \geq 2$,
\begin{equation} \label{eq:est_x2p}
    \Ex\left\{ \| \qx \|^{2p}\right\}  = O(1).
\end{equation}
\end{Lemma}

\begin{proof}
We will frequently use the fact that if $|X_i| \leq \varepsilon\sqrt{N}$, then $\Ex \{|X_i|^3\} \leq \varepsilon\sqrt{N}~\Ex\{ |X_i|^2 \} = \varepsilon\sqrt{N} $
and more generally, $\Ex \{|X_i|^p\} \leq \left(\varepsilon\sqrt{N}\right)^{p-2} \Ex \{ |X_i|^2 \} = \left(\varepsilon\sqrt{N}\right)^{p-2}$.

We first prove (\ref{eq:est_xc2p}). Using Lemma \ref{lemma_eleIneq}, we have, for $p \geq 1$,
\begin{align*}
    \Ex \left\{| \qc^H \qx |^{2p}\right\}  = \Ex\left\{  \left| \frac{1}{\sqrt{N}} \sum_{i=1}^{N} c_i X_i \right|^{2p}\right\}
                       &\leq \frac{C_p}{N^p} \left( \sum_{i=1}^{N} |c_i|^{2p} \Ex \left\{|X_i|^{2p}\right\} + \left(\sum_{i=1}^{N} |c_i|^2 \Ex \left\{|X_i|^{2}\right\}\right)^p \right) \\
                       &\leq \frac{C_p}{N^p} \left( C  \left(\varepsilon\sqrt{N}\right)^{2p-2} + C \right) \\
                       & = O\left(\frac{1}{N}\right).
\end{align*}

Next, we prove (\ref{eq:est_x2p}). For any $p \geq 2$, we have
\begin{align*}
    \Ex \left\{\| \qx \|^{2p}\right\}  = \Ex\left\{ | \qx^H \qx |^{p}\right\}
                       &\leq C \left( \Ex\left\{ \left| \qx^H \qx - \Ex \left\{\qx^H \qx\right\} \right|^{p}\right\} \right) + C \left(\Ex \left\{\qx^H \qx\right\} \right)^{p}  \\
                       &\leq \frac{C C_p}{N^p} \left( \sum_{i=1}^{N} \Ex \left\{ \left| |X_i|^2 - \Ex\left\{|X_i|^2\right\} \right|^{p}\right\} + \left(\sum_{i=1}^{N} \Ex \left\{ \left| |X_i|^2 - \Ex\left\{|X_i|^2 \right\} \right|^{2}\right\}\right)^\frac{p}{2} \right) + C \\
                       &\leq \frac{C' C_p}{N^p} \left(N \cdot \left(\varepsilon\sqrt{N}\right)^{2p-2} + \left(N \cdot \varepsilon^2 N\right)^\frac{p}{2} \right) + C \\
                       & = O(1),
\end{align*}
where the first inequality follows from Lemma \ref{sumIneq} and the second inequality follows from Lemma \ref{lemma_eleIneq}.
\end{proof}

\begin{Lemma} \label{traceIneq}
{\rm (\cite[Lemma B.26]{Bai-10}.)} Let $\qA \in \bbC^{N \times N}$ be a nonrandom matrix and $\qx = [X_i] \in \bbC^{N}$ be a random vector of independent entries.
Assume that $\Ex\{X_i\} = 0$, $\Ex\{|X_i|^2\} = 1$, and $\Ex\{|X_i|^l\} \leq c_l$. Then, for any $p \geq 1$,
\begin{equation*}
    \Ex\left\{ |\qx^H\qA\qx - \tr(\qA)|^p \right\} \leq C_p \left( \left(c_4 \tr(\qA\qA^H)\right)^{\frac{p}{2}} + c_{2p} \tr\left((\qA\qA^H)^{\frac{p}{2}}\right) \right).
\end{equation*}
\end{Lemma}

\begin{Lemma} {\rm \cite[0.7.3]{Hor-91}} \label{invLemma}
Let $\qA$ be partitioned as
\[
    \qA = \left[
    \begin{array}{cc}
     \qA_{11} & \qA_{12} \\
     \qA_{21} & \qA_{22}
    \end{array}
    \right].
\]
Invertibility is assumed for any subblock whose inverse is indicated. Then,
\[
    \qA^{-1} = \left[
    \begin{array}{cc}
     (\qA_{11}-\qA_{12}\qA_{22}^{-1}\qA_{21})^{-1} & \qA_{11}^{-1}\qA_{21}(\qA_{21}\qA_{11}^{-1}\qA_{12}-\qA_{22})^{-1} \\
     (\qA_{21}\qA_{11}^{-1}\qA_{12}-\qA_{22})^{-1}\qA_{21}\qA_{11}^{-1} & (\qA_{22}-\qA_{21}\qA_{11}^{-1}\qA_{12})^{-1}
    \end{array}
    \right].
\]
\end{Lemma}

\begin{Lemma} \label{lineIneqLemma}
{\rm \cite[Lemma 9]{Cou-09}} If the components of $\qC$, $\qx$, and $\qb$ are all positive, then $\qx = \qC\qx + \qb$ implies $\rho(\qC) < 1$.
\end{Lemma}

\begin{Lemma} \label{matrixIneqLemma1}
{\rm \cite[Theorem 8.1.18]{Hor-85}} Let $\qA = [A_{ij}]$ and $\qB = [B_{ij}]$ be square matrices. If $|A_{ij}| \leq B_{ij}, ~\forall i,j$, then $\rho(\qA) \leq
\rho(|\qA|) \leq \rho(\qB)$.
\end{Lemma}

\begin{Lemma} \label{matrixIneqLemma2}
{\rm (\cite[Lemma 5.7.9]{Hor-91}} Let $\qA = [A_{ij}]$  and $\qB = [B_{ij}]$ be matrices with nonnegative elements. Then $\rho\left(
\left[A_{ij}^\frac{1}{2}B_{ij}^\frac{1}{2}\right] \right) \leq \rho(\qA)^\frac{1}{2} \rho(\qB)^\frac{1}{2}$.
\end{Lemma}

\begin{Lemma} \label{matrixIneqLemma3}
Let $\qA$ and $\qB$ be any matrices such that $\qA\qB^H$ exists and is a squared matrix. If $\tr(\qA\qA^H) \leq 1$ and $\tr(\qB\qB^H) \leq 1$, then
\begin{equation*}
     \left| 1 - \tr(\qA\qB^H) \right|
     \geq
    \left(1 - \tr(\qA\qA^H) \right)^\frac{1}{2}
    \left(1 - \tr(\qB\qB^H) \right)^\frac{1}{2}.
\end{equation*}
\end{Lemma}

\begin{proof}
For real numbers $a$ and $b$ with $a,b \in [0, 1]$ it is easily shown that
\begin{equation} \label{eq:simIneq}
    (1-a)^\frac{1}{2}(1-b)^\frac{1}{2} \leq 1-\sqrt{ab}.
\end{equation}
Let $a = \tr(\qA\qA^H)$ and $b= \tr(\qB\qB^H)$. Plugging $a$ and $b$ into (\ref{eq:simIneq}), we obtain
\begin{align*}
    \left(1-\tr(\qA\qA^H)\right)^\frac{1}{2}\left(1-\tr(\qB\qB^H)\right)^\frac{1}{2}
    \leq 1- \sqrt{\tr(\qA\qA^H) \tr(\qB\qB^H)}
    \leq 1- \left|\tr(\qA\qB^H) \right|
    \leq \left|1- \tr(\qA\qB^H) \right|,
\end{align*}
where the second inequality follows from Lemma \ref{LemmaNorm}.
\end{proof}

\begin{Lemma} {\rm \cite[Proposition 2.2]{Hac-07}} \label{StjLemma1}
Let $\vartheta(z) \in \bbS(\bbR^+)$ with $\mu$ being its associated measure carried by $\bbR^+$. We have the following results:
\begin{enumerate}
\item[1)] $\vartheta(z)$ is analytic on $\bbC-\bbR^+$;
\item[2)] $\Imag{\vartheta(z)} > 0$ if $\Imag{z} > 0$, and $\Imag{\vartheta(z)} < 0$ if $\Imag{z} < 0$;
\item[3)] $\Imag{z\vartheta(z)} > 0$ if $\Imag{z} > 0$, and $\Imag{z\vartheta(z)} < 0$ if $\Imag{z} < 0$;
\item[4)] $\mu(\bbR^+) = \lim_{y \rightarrow \infty} -\sfj y ~\vartheta(\sfj y)$.
\end{enumerate}
\end{Lemma}

{\renewcommand{\baselinestretch}{1.25}
\begin{footnotesize}

\end{footnotesize}}

\end{document}